\documentclass[12pt,onecolumn,draftclsnofoot]{IEEEtran}
\usepackage{fancyhdr}
\usepackage{amsmath,epsfig}
\usepackage{threeparttable}
\usepackage{epsf,epsfig}
\usepackage{amsmath}
\usepackage{amssymb}
\usepackage{amsfonts}
\usepackage{algorithmic}
\usepackage{algorithm}
\usepackage[noadjust]{cite}
\usepackage{dsfont}
\usepackage{subfigure}

\begin{document}
\newtheorem{theorem}{Theorem}
\newtheorem{acknowledgement}[theorem]{Acknowledgement}
\newtheorem{axiom}[theorem]{Axiom}
\newtheorem{case}[theorem]{Case}
\newtheorem{claim}[theorem]{Claim}
\newtheorem{conclusion}[theorem]{Conclusion}
\newtheorem{condition}[theorem]{Condition}
\newtheorem{conjecture}[theorem]{Conjecture}
\newtheorem{criterion}[theorem]{Criterion}
\newtheorem{definition}[theorem]{Definition}
\newtheorem{example}[theorem]{Example}
\newtheorem{exercise}[theorem]{Exercise}
\newtheorem{lemma}{Lemma}
\newtheorem{corollary}{Corollary}
\newtheorem{notation}[theorem]{Notation}
\newtheorem{problem}[theorem]{Problem}
\newtheorem{proposition}{Proposition}
\newtheorem{solution}[theorem]{Solution}
\newtheorem{summary}[theorem]{Summary}
\newtheorem{assumption}{Assumption}
\newtheorem{examp}{\bf Example}
\newtheorem{probform}{\bf Problem}
\def\remark{{\noindent \bf Remark:\hspace{0.5em}}}

\def\qed{$\Box$}
\def\QED{\mbox{\phantom{m}}\nolinebreak\hfill$\,\Box$}
\def\proof{\noindent{\emph{Proof:} }}
\def\poof{\noindent{\emph{Sketch of Proof:} }}
\def
\endproof{\hspace*{\fill}~\qed
\par
\endtrivlist\unskip}
\def\endproof{\hspace*{\fill}~\qed\par\endtrivlist\vskip3pt}

\def\E{\mathbb{E}}
\def\eps{\varepsilon}
\def\phi{\varphi}
\def\Lsp{{\boldsymbol L}}
\def\Bsp{{\boldsymbol B}}
\def\lsp{{\boldsymbol\ell}}
\def\Ltsp{{\Lsp^2}}
\def\Lpsp{{\Lsp^p}}
\def\Linsp{{\Lsp^{\infty}}}
\def\LtR{{\Lsp^2(\Rst)}}
\def\ltZ{{\lsp^2(\Zst)}}
\def\ltsp{{\lsp^2}}
\def\ltZt{{\lsp^2(\Zst^{2})}}
\def\ninN{{n{\in}\Nst}}
\def\oh{{\frac{1}{2}}}
\def\grass{{\cal G}}
\def\ord{{\cal O}}
\def\dist{{d_G}}
\def\conj#1{{\overline#1}}
\def\ntoinf{{n \rightarrow \infty }}
\def\toinf{{\rightarrow \infty }}
\def\tozero{{\rightarrow 0 }}
\def\trace{{\operatorname{trace}}}
\def\ord{{\cal O}}
\def\UU{{\cal U}}
\def\rank{{\operatorname{rank}}}
\def\acos{{\operatorname{acos}}}

\def\SINR{\mathsf{SINR}}
\def\SNR{\mathsf{SNR}}
\def\SIR{\mathsf{SIR}}
\def\tSIR{\widetilde{\mathsf{SIR}}}
\def\Ei{\mathsf{Ei}}
\def\l{\left}
\def\r{\right}
\def\({\left(}
\def\){\right)}
\def\lb{\left\{}
\def\rb{\right\}}

\setcounter{page}{1}

\newcommand{\eref}[1]{(\ref{#1})}
\newcommand{\fig}[1]{Fig.\ \ref{#1}}

\def\bydef{:=}
\def\ba{{\mathbf{a}}}
\def\bb{{\mathbf{b}}}
\def\bc{{\mathbf{c}}}
\def\bd{{\mathbf{d}}}
\def\bee{{\mathbf{e}}}
\def\bff{{\mathbf{f}}}
\def\bg{{\mathbf{g}}}
\def\bh{{\mathbf{h}}}
\def\bi{{\mathbf{i}}}
\def\bj{{\mathbf{j}}}
\def\bk{{\mathbf{k}}}
\def\bl{{\mathbf{l}}}
\def\bm{{\mathbf{m}}}
\def\bn{{\mathbf{n}}}
\def\bo{{\mathbf{o}}}
\def\bp{{\mathbf{p}}}
\def\bq{{\mathbf{q}}}
\def\br{{\mathbf{r}}}
\def\bs{{\mathbf{s}}}
\def\bt{{\mathbf{t}}}
\def\bu{{\mathbf{u}}}
\def\bv{{\mathbf{v}}}
\def\bw{{\mathbf{w}}}
\def\bx{{\mathbf{x}}}
\def\by{{\mathbf{y}}}
\def\bz{{\mathbf{z}}}
\def\b0{{\mathbf{0}}}

\def\bA{{\mathbf{A}}}
\def\bB{{\mathbf{B}}}
\def\bC{{\mathbf{C}}}
\def\bD{{\mathbf{D}}}
\def\bE{{\mathbf{E}}}
\def\bF{{\mathbf{F}}}
\def\bG{{\mathbf{G}}}
\def\bH{{\mathbf{H}}}
\def\bI{{\mathbf{I}}}
\def\bJ{{\mathbf{J}}}
\def\bK{{\mathbf{K}}}
\def\bL{{\mathbf{L}}}
\def\bM{{\mathbf{M}}}
\def\bN{{\mathbf{N}}}
\def\bO{{\mathbf{O}}}
\def\bP{{\mathbf{P}}}
\def\bQ{{\mathbf{Q}}}
\def\bR{{\mathbf{R}}}
\def\bS{{\mathbf{S}}}
\def\bT{{\mathbf{T}}}
\def\bU{{\mathbf{U}}}
\def\bV{{\mathbf{V}}}
\def\bW{{\mathbf{W}}}
\def\bX{{\mathbf{X}}}
\def\bY{{\mathbf{Y}}}
\def\bZ{{\mathbf{Z}}}

\def\mA{{\mathbb{A}}}
\def\mB{{\mathbb{B}}}
\def\mC{{\mathbb{C}}}
\def\mD{{\mathbb{D}}}
\def\mE{{\mathbb{E}}}
\def\mF{{\mathbb{F}}}
\def\mG{{\mathbb{G}}}
\def\mH{{\mathbb{H}}}
\def\mI{{\mathbb{I}}}
\def\mJ{{\mathbb{J}}}
\def\mK{{\mathbb{K}}}
\def\mL{{\mathbb{L}}}
\def\mM{{\mathbb{M}}}
\def\mN{{\mathbb{N}}}
\def\mO{{\mathbb{O}}}
\def\mP{{\mathbb{P}}}
\def\mQ{{\mathbb{Q}}}
\def\mR{{\mathbb{R}}}
\def\mS{{\mathbb{S}}}
\def\mT{{\mathbb{T}}}
\def\mU{{\mathbb{U}}}
\def\mV{{\mathbb{V}}}
\def\mW{{\mathbb{W}}}
\def\mX{{\mathbb{X}}}
\def\mY{{\mathbb{Y}}}
\def\mZ{{\mathbb{Z}}}

\def\cA{\mathcal{A}}
\def\cB{\mathcal{B}}
\def\cC{\mathcal{C}}
\def\cD{\mathcal{D}}
\def\cE{\mathcal{E}}
\def\cF{\mathcal{F}}
\def\cG{\mathcal{G}}
\def\cH{\mathcal{H}}
\def\cI{\mathcal{I}}
\def\cJ{\mathcal{J}}
\def\cK{\mathcal{K}}
\def\cL{\mathcal{L}}
\def\cM{\mathcal{M}}
\def\cN{\mathcal{N}}
\def\cO{\mathcal{O}}
\def\cP{\mathcal{P}}
\def\cQ{\mathcal{Q}}
\def\cR{\mathcal{R}}
\def\cS{\mathcal{S}}
\def\cT{\mathcal{T}}
\def\cU{\mathcal{U}}
\def\cV{\mathcal{V}}
\def\cW{\mathcal{W}}
\def\cX{\mathcal{X}}
\def\cY{\mathcal{Y}}
\def\cZ{\mathcal{Z}}
\def\cd{\mathcal{d}}
\def\Mt{M_{t}}
\def\Mr{M_{r}}
\def\O{\Omega_{M_{t}}}
\newcommand{\figref}[1]{{Fig.}~\ref{#1}}
\newcommand{\tabref}[1]{{Table}~\ref{#1}}

\newcommand{\var}{\mathsf{var}}
\newcommand{\fb}{\tx{fb}}
\newcommand{\nf}{\tx{nf}}
\newcommand{\BC}{\tx{(bc)}}
\newcommand{\MAC}{\tx{(mac)}}
\newcommand{\Pout}{P_{\mathsf{out}}}
\newcommand{\nnn}{\nn\\}
\newcommand{\FB}{\tx{FB}}
\newcommand{\TX}{\tx{TX}}
\newcommand{\RX}{\tx{RX}}
\renewcommand{\mod}{\tx{mod}}
\newcommand{\m}[1]{\mathbf{#1}}
\newcommand{\td}[1]{\tilde{#1}}
\newcommand{\sbf}[1]{\scriptsize{\textbf{#1}}}
\newcommand{\stxt}[1]{\scriptsize{\textrm{#1}}}
\newcommand{\suml}[2]{\sum\limits_{#1}^{#2}}
\newcommand{\sumlk}{\sum\limits_{k=0}^{K-1}}
\newcommand{\eqhsp}{\hspace{10 pt}}
\newcommand{\tx}[1]{\texttt{#1}}
\newcommand{\Hz}{\ \tx{Hz}}
\newcommand{\sinc}{\tx{sinc}}
\newcommand{\tr}{\mathrm{tr}}
\newcommand{\diag}{\mathrm{diag}}
\newcommand{\MAI}{\tx{MAI}}
\newcommand{\ISI}{\tx{ISI}}
\newcommand{\IBI}{\tx{IBI}}
\newcommand{\CN}{\tx{CN}}
\newcommand{\CP}{\tx{CP}}
\newcommand{\ZP}{\tx{ZP}}
\newcommand{\ZF}{\tx{ZF}}
\newcommand{\SP}{\tx{SP}}
\newcommand{\MMSE}{\tx{MMSE}}
\newcommand{\MINF}{\tx{MINF}}
\newcommand{\RC}{\tx{MP}}
\newcommand{\MBER}{\tx{MBER}}
\newcommand{\MSNR}{\tx{MSNR}}
\newcommand{\MCAP}{\tx{MCAP}}
\newcommand{\vol}{\tx{vol}}
\newcommand{\ah}{\hat{g}}
\newcommand{\tg}{\tilde{g}}
\newcommand{\teta}{\tilde{\eta}}
\newcommand{\heta}{\hat{\eta}}
\newcommand{\uh}{\m{\hat{s}}}
\newcommand{\eh}{\m{\hat{\eta}}}
\newcommand{\hv}{\m{h}}
\newcommand{\hh}{\m{\hat{h}}}
\newcommand{\Po}{P_{\mathrm{out}}}
\newcommand{\Poh}{\hat{P}_{\mathrm{out}}}
\newcommand{\Ph}{\hat{\gamma}}
\newcommand{\mat}[1]{\begin{matrix}#1\end{matrix}}
\newcommand{\ud}{^{\dagger}}
\newcommand{\C}{\mathcal{C}}
\newcommand{\nn}{\nonumber}
\newcommand{\nInf}{U\rightarrow \infty}

\title{\huge \setlength{\baselineskip}{30pt} Stability and Delay of Zero-Forcing SDMA with Limited Feedback}
\author{Kaibin Huang and Vincent K. N. Lau \thanks{\setlength{\baselineskip}{15pt} K. Huang and V. K. N. Lau are with Department of Electronic and Computer Engineering, Hong Kong University of Science and Technology, Clear Water Bay, Kowloon, Hong Kong. Email: huangkb@ieee.org, eeknlau@ee.ust.hk}}

\maketitle
\vspace{-30pt}
\begin{abstract}\setlength{\baselineskip}{15pt}
This paper addresses the stability and queueing delay of Space Division Multiple Access (SDMA) systems with bursty traffic, where zero-forcing beamforming enables simultaneous transmission to multiple mobiles. Computing beamforming vectors relies on quantized channel state information (CSI) feedback (limited feedback) from mobiles. Define the \emph{stability region} for SDMA as the set of multiuser packet-arrival rates for which the steady-state queue lengths are finite. Given perfect CSI feedback  and equal power allocation over scheduled queues, the stability region is proved to be a convex polytope having the derived vertices. For any set of arrival rates in the stability region, multiuser queues are shown to be stabilized by a joint queue-and-beamforming control policy that maximizes the departure-rate-weighted sum of queue lengths. The stability region for limited feedback is found to be the perfect-CSI region multiplied by one minus a small factor. The required number of feedback bits per mobile is proved to scale logarithmically with the inverse of the above factor as well as linearly with the number of transmit antennas minus one. The effects of limited feedback on queueing delay are also quantified. For Poisson arrival processes, CSI quantization errors are shown to multiply average queueing delay by a factor $M>1$. For given $M\rightarrow 1$, the number of feedback bits per mobile $B$ is proved to be $O(-\log_2(1-1/M))$. For general arrival processes, CSI errors are found to increase Kingman's bound on the tail probability of the instantaneous delay by one plus a small factor $\eta$. For given $\eta\rightarrow 0$, $B$ is proved to be $O(-\log_2\eta)$.
\end{abstract}

\section{Introduction}\label{Section:Intro}
In this paper, we consider a Space Division Multiple Access (SDMA) system where a multi-antenna base station transmits simultaneously to multiple single-antenna mobiles. Given feedback of channel state information (CSI) from mobiles, data packets transmitted using SDMA are decoupled by zero-forcing beamforming at the base station. To design such a system, it is important to understand the maximum packet arrival rates the system can support without becoming unstable, and the effects of the practical finite-rate CSI feedback (limited feedback) on system stability and queueing delay. We address these issues by deriving the \emph{stability region} for zero-forcing SDMA, defined as the set of packet arrival rates for which finite steady-state queue lengths are feasible. We also characterize  the required amount of limited feedback overhead for constraining the degradation of system stability and delay performance due to limited feedback. These results provide insight into designing admission and queue control as well as limited feedback in multi-antenna SDMA systems.

\subsection{Prior Work and Motivation}
SDMA for multi-antenna downlink systems has emerged as a key technology for enabling high-rate wireless access \cite{Gesbert:ShiftMIMOParadigm:2007}. The key feature of SDMA is the support of multiuser data streams by exploiting the spatial degrees of freedom. This feature can be realized by pre-cancellation of multiuser interference using the optimal \emph{dirty paper coding} (DPC) \cite{Caire:AchivThroghputBroadcastChan:2003, Vishwanath:DualityAchRatesBroadcastChan:2003}. This technique not only has high complexity but also requires non-causal and perfect CSI at the transmitter (CSIT). For these reasons, multiuser beamforming such as zero-forcing beamforming has become a popular alternative solution by providing low complexity and close-to-optimal performance \cite{YooGoldsmith:OptimBroadcastZeroForcingBeam:2006, Jindal:MIMOBroadcastFiniteRateFeedback:06, Huang:OrthBeamSDMALimtFb:07, SpencerSwindleETAL:ZFsdma:2004}. Furthermore, multiuser beamforming admits efficient scheduling \cite{SwannackWornell:FindingNEMO:2005, YooJindal:FiniteRateBroadcastMUDiv:2007, SharifHassibi:CapMIMOBroadcastPartSideInfo:Feb:05}.

In practice, feedback CSI for multiuser beamforming must be quantized given finite-rate feedback channels \cite{Love:OverviewLimitFbWirelssComm:2008}.  Limited feedback introduces interference between multiuser data streams even if zero-forcing beamforming is applied \cite{YooJindal:FiniteRateBroadcastMUDiv:2007, Jindal:MIMOBroadcastFiniteRateFeedback:06}. The resultant throughput loss can be controlled by adjusting the number of CSI feedback bits. In particular, for zero-forcing beamforming, this number must scale with the number of transmit antennas and the signal-to-noise ratio (SNR) to constrain the throughput loss \cite{Jindal:MIMOBroadcastFiniteRateFeedback:06}. If the user pool is large, the effects of CSIT quantization errors can be alleviated by exploiting multiuser diversity \cite{SwannackWorenell:BroadcastChanLimitedCSI:2005, SharifHassibi:CapMIMOBroadcastPartSideInfo:Feb:05, Huang:OrthBeamSDMALimtFb:07, YooJindal:FiniteRateBroadcastMUDiv:2007}. The above works on SDMA and limited feedback assume mobiles with infinite back-logged data. However, it is important to address the issue of data burstiness for SDMA supporting delay-sensitive applications.

Active research has been conducted on multiuser systems with bursty traffic by integrating information and queueing theory. In \cite{Telatar:CombineQueueTheoInfoTheory:1995}, the delay and throughput of multiple-access channels with Poisson packet arrivals are studied as a processor-sharing problem, where data packets are decoded by treating multiuser interference as noise. In \cite{YehThesis:MACFadingCommNeworks:2001, YehCohen:ThputDelayOptimalResourceAlloacMAC:2003}, the more sophisticated successive decoding \cite{CoverBook} is applied in multiple access channels with bursty data. These studies show the delay and stability optimality of the queue control policy that allocates the highest data rates to mobiles with the longest queues, aligned with the well-known principle of \emph{max weight match} \cite{GeorgNeelyBook}. Given successive decoding, the stability regions for both the multiple-access and the broadcast channels are found to be identical to their corresponding information capacity regions \cite{YehCohen:ThputOptimPowrRateControlMACBC:2004}. A general broadcast system of multiple queues and multiple servers is studied in \cite{Neely:PowrAllocationRoutMultibeamSatellites:2003}, where the max-weight-match policy for resource and server allocation is shown to be stability optimal. These works assume perfect CSIT available at the base station. The issues of stability and queueing delay for SDMA with limited feedback are challenging  and still unexplored.

Existing work on multi-antenna SDMA with bursty data  focuses on enlarging  the stability region by exploiting multiuser diversity with perfect CSIT \cite{SwannackWorenell:LowComplexMuScheduleMIMOBC:2004, VisKum:RateScheduleMultiAntDLSys:2005}. Based on the max-weight-match principle, these scheduling algorithms select mobiles with long queue lengths and whose channel vectors are \emph{nearly-orthogonal} with each other. The orthogonality criterion reduces multiuser interference \cite{VisKum:RateScheduleMultiAntDLSys:2005} or signal power loss due to interference avoidance \cite{SwannackWorenell:LowComplexMuScheduleMIMOBC:2004}. However, the existing algorithms are ineffective for systems with small numbers of mobiles, or those supporting delay-sensitive applications where scheduling based on mobiles' CSI is infeasible \cite{Jindal:MIMOBroadcastFiniteRateFeedback:06}. For these systems, we propose the approach of integrating queue control and adaptive beamforming. Thereby the downlink stability region is enlarged by dynamic control of the number of active queues and spatial streams.

\subsection{Contributions and Organization}
We consider a narrow-band SDMA system where a multi-antenna base station serves multiple single-antenna mobiles using zero-forcing beamforming with limited feedback. The number of mobiles is equal to that  of transmit antennas. The channel coefficients are i.i.d. $\mathcal{CN}(0,1)$ random variables. The  data packets for different mobiles are stored in separate buffers of infinite capacity at the base station, forming multiple queues. In each time slot, a subset of queues are scheduled and corresponding data packets are transmitted using zero-forcing beamforming and \emph{automatic repeat request} (ARQ) \cite{BertsekasBook:DataNetwk:92}. Packet transmission is successful if the received signal-to-interference-and-noise (SINR) exceeds a given threshold.

We investigate the stability and delay performance of the above system. Our main contributions are summarized as follows. Given perfect CSIT and equal power allocation over scheduled queues, the SDMA stability region is proved to be a convex polytope having the derived vertices. Each vertex  corresponds to scheduling a unique subset of queues for SDMA.  Queues with arrival rates in the stability region are shown to be stabilized by the queue-and-beamforming control policy that maximizes the departure-rate-weighted sum of queue lengths. In addition, we obtain the stability region for flexible power control. Next, we analyze the feedback requirements for zero-forcing forming with limited feedback. The required number of feedback bits  per user $B$ is derived under the constraint that the stability region is no smaller than the perfect-CSI counterpart scaled by one minus a small factor. Specifically, $B$ is shown to increase logarithmically with this factor as well as linearly with the number of transmit antennas minus one. The feedback requirements are also analyzed using the criterion of queueing delay. For Poisson arrival processes, $B$ is obtained for bounding the average-delay ratio between limited and perfect CSI feedback. In particular, $B$ is shown to have the same order as the logarithm of the above ratio if it is close to one. For general arrival processes, the complementary cumulative distribution function (CCDF) of instantaneous delay is  upper bounded by Kingman's bound \cite{GallagerBook:StochasticProcs:95}. Using \emph{perturbation theory} \cite{SimBook:FirstLookPerturbationTheory:97}, $B$ is derived given a constraint on the multiplier of the CCDF bound caused by CSIT inaccuracy.

The remainder of this paper is organized as follows. The model of the SDMA system is described in Section~\ref{Section:Sys}. In Section~\ref{Section:PerfCSI}, assuming perfect CSI feedback, we discuss the stability region for SDMA and the joint queue-and-beamforming control policy for stabilizing the SDMA system. The CSI feedback requirements for constraining the stability and delay performance loss due to imperfect CSI are analyzed in Section~\ref{Sectoin:FbRequire}. Numerical results are presented in Section~\ref{Section:Simualtion}, followed by concluding remarks in Section~\ref{Section:Conlusion}.

\section{System Description} \label{Section:Sys}
We consider a narrow-band SDMA system where a base station with  $L$ antennas serves $L$ single-antenna mobiles as illustrated Fig.~\ref{Fig:MIMOBC}.  Following \cite{Tel:CapaMultGausChan:99, Jindal:MIMOBroadcastFiniteRateFeedback:06, TarokhJafETAL:SpacBlocCodeFrom:Jul:99}, channel coefficients are modeled as i.i.d. $\mathcal{CN}(0,1)$ random variables, corresponding to a rich scattering environment. Channel variation over time is modeled as \emph{block fading}, namely that each realization of channel coefficients is fixed within one time slot and different realizations are independent \cite{CaireETAL:PwrCtrlFadingChan:99}. Without loss of generality, each time slot is assumed to span one time unit. Moreover, the fading process is assumed stationary and ergodic. The vector channel for the $\ell$th user is represented by a $L\times 1$ random vector $\bh_\ell(t)$ where $t$ is the slot index.

\begin{figure}
\centering
  \includegraphics[width=14cm]{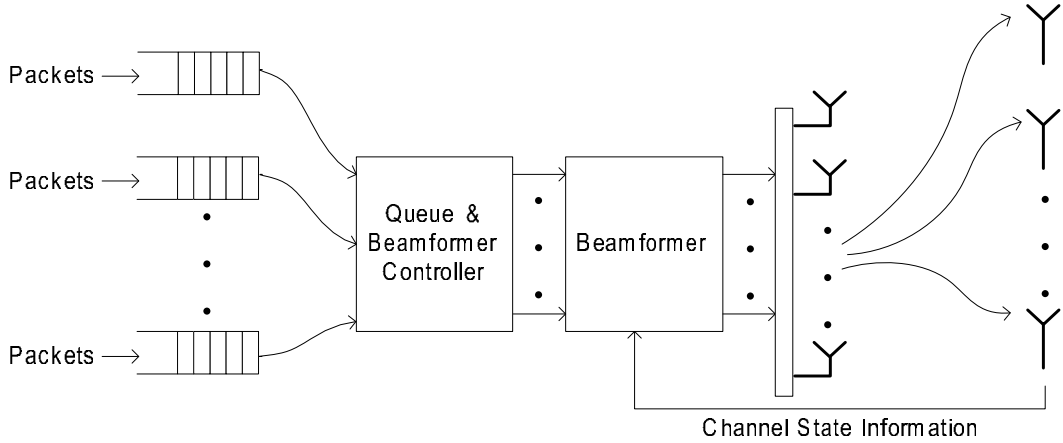}\\
  \caption{Multi-antenna SDMA system}\label{Fig:MIMOBC}
\end{figure}

\subsection{Data Source and Transmission Protocol}\label{Section:Data}
The data for each user is generated in packet. For simplicity, the packet length is assumed fixed and equal to one time unit. Packets for different mobiles arrive \emph{asynchronously} at separate buffers with infinite capacity, forming $L$ queues at the base station. Multiuser packet arrival rates may be unequal, modeling multimedia applications. At the beginning of each slot, packets from scheduled queues are transmitted using zero-forcing beamforming  (see Section~\ref{Section:Beamform}). Packet transmission uses the following ARQ protocol \cite{BertsekasBook:DataNetwk:92}. A packet is received successfully if the SINR at the intended mobile exceeds a given  threshold $\theta$. Otherwise the packet is retransmitted in the next scheduled time slot. Each mobile reports to the base station the status of each received packet via perfect ACK/NACK feedback.

\subsection{Zero-Forcing Beamforming}\label{Section:Beamform}
Consider a set of $K$ nonempty queues scheduled for transmission in an arbitrary time slot, which are specified by the index set $\mathcal{A}$. For simplicity, we assume equal allocation of the total transmission power $P$ over scheduled queues, and the effect of relaxing this assumption on system stability is discussed in Section~\ref{Section:PowerControl}. The $K$ transmitted packets are separated using beamforming vectors represented by $\{\bff_\ell: \ell \in\mathcal{A}\}$. Let $X_\ell := \mathcal{CN}(0,1)$ and $Z_\ell$ represent the transmitted and received data symbols for the $\ell$th user, respectively, where the symbol index is omitted for brevity. We can write $Z_\ell$ as
\begin{equation}\label{Eq:RxSymb}
Z_\ell = \gamma\bff_\ell^\dagger\bh_\ell X_\ell + \gamma\sum\nolimits_{\substack{m\in\mathcal{A}\\m\neq \ell}} \bff_m^\dagger\bh_\ell X_m + N_\ell,\quad \forall \ \ell \in \mathcal{A}
\end{equation}
where $N_\ell:=\mathcal{CN}(0,1)$ denotes a sample of the additive white Gaussian noise process, and $\gamma:=\frac{P}{K}$ is the signal-to-noise ratio (SNR). Note that the summation term in \eqref{Eq:RxSymb} represents multi-user interference.
Let $\hat{\bh}_\ell$ denote the feedback CSI generated from $\bh_\ell$. The base station computes beamforming vectors using $\{\hat{\bh}_\ell: \ell \in\mathcal{A}\}$ and the zero-forcing method \cite{Jindal:MIMOBroadcastFiniteRateFeedback:06}. To be precise, $\bff_\ell$ is chosen to be orthogonal to the channel subspace $\mathcal{H}_\ell$ spanned by the vectors $\{\hat{\bh}_m\mid m\neq \ell, m\in\mathcal{A}\}$ \cite{YooJindal:FiniteRateBroadcastMUDiv:2007, Jindal:MIMOBroadcastFiniteRateFeedback:06}; under this constraint,  $\bff_\ell$  maximizes the received SNR, namely $\bff_\ell = \arg\max_{\bff\in\mathcal{H}^c_\ell} |\bff_\ell^\dagger\bh_\ell|^2$ where $\mathcal{H}^c_\ell$ denotes the null space of $\mathcal{H}$ \cite{Jindal:RethinkMIMONetwork:LinearThroughput:2008}.
For perfect CSIT, these zero-forcing constraint nulls multi-user interference
(the summation term in \eqref{Eq:RxSymb} is equal to zero) for the case of perfect CSIT $(\hat{\bh}_\ell = \bh_\ell)$ or otherwise residual interference exists. Thus the SNR/SINR can be written from \eqref{Eq:RxSymb} as
\begin{equation}\label{Eq:SINR:Def}
\l\{
\begin{aligned}
&\SNR_\ell = \gamma|\bff_\ell^\dagger\bh_\ell|^2&&\textrm{perfect CSI feedback}\\
&\SINR_\ell = \frac{\gamma|\bff_\ell^\dagger\bh_\ell|^2}{1 + \gamma\sum_{\substack{m\in\mathcal{A}\\m\neq \ell}} |\bff_m^\dagger \bh_\ell|^2}&&\textrm{limited feedback}.
\end{aligned}
\r.
\end{equation}
Zero-forcing beamforming with perfect CSIT and limited feedback are considered in Section~\ref{Section:PerfCSI} and Section~\ref{Sectoin:FbRequire}, respectively.

\subsection{Limited Feedback}\label{Section:LimFb}
In practice, zero-forcing beamforming relies on limited feedback from mobiles.  Specifically, each mobile estimates the corresponding vector channel using pilot signals broadcast by the base station. Then all mobiles send back CSI via finite-rate feedback channels to the base station for enabling multiuser beamforming. We assume perfect CSI estimation at mobiles and zero delay and error for the feedback channel following the literature (see e.g. \cite{LovHeaETAL:Gras:May:03, Jindal:MIMOBroadcastFiniteRateFeedback:06}). Thereby the CSIT inaccuracy is contributed entirely by CSI quantization.  To simplify analysis, we adopt the following quantized CSI model  from \cite{YooJindal:FiniteRateBroadcastMUDiv:2007, Zhou:QuantifyPowrLossTxBeamFiniteRateFb:2005}. Let a unitary vector $\hat{\bh}_\ell$ represent the output of quantizing $\bh_\ell$. Define the quantization error as $\epsilon_\ell = 1-|\hat{\bh}_\ell^\dagger\bh_\ell|^2/\|\bh_\ell\|^2$. The model in \cite{YooJindal:FiniteRateBroadcastMUDiv:2007, Zhou:QuantifyPowrLossTxBeamFiniteRateFb:2005} approximates a Voronoi cell \cite{GerGra:VectQuanSignComp:92} as a sphere cap, or more precisely $\hat{\bh}_\ell$  is isotropic in the set (Voronoi cell) $\mathcal{V} = \l\{\bs \in \mathds{C}^L: \|\bs\|=1, \frac{|\bs^\dagger \bh_\ell|^2}{\|\bh_\ell\|^2}\leq 1-2^{-\frac{B}{L-1}}\r\}$. As a result, the CDF of $\epsilon_\ell$ is \cite{YooJindal:FiniteRateBroadcastMUDiv:2007}
\begin{equation}\label{Eq:QuantErr}
\Pr(\epsilon_\ell \leq a) = 2^B a^{L-1},\quad 0\leq a \leq 2^{-\frac{B}{L-1}}.
\end{equation}
The above model is observed in \cite{YooJindal:FiniteRateBroadcastMUDiv:2007, Zhou:QuantifyPowrLossTxBeamFiniteRateFb:2005} to be accurate for the practical codebook-based CSI quantization for beamforming \cite{LovHeaETAL:GrasBeamMultMult:Oct:03}.

\subsection{Joint Queue-and-Beamforming Control}\label{Section:Sys:Control}
In each slot, the controller (see Fig.~\ref{Fig:MIMOBC}) schedules a subset of nonempty queues for SDMA transmission. Furthermore, the controller informs the beamformer to create the matching spatial streams for transmitting packets from scheduled queues. For simplicity, we consider equal transmission power for scheduled queues. This assumption is relaxed in Section~\ref{Section:PowerControl}. Under this assumption, the controller's decision in the $t$th slot can be represented by an $L\times 1$ indicator vector $\bm(t)\in \mathcal{V}:=\{0, 1\}^L$, where the $\ell$th component of $\bm(t)$ is equal to one if the $\ell$th queue is scheduled or otherwise equal to $0$. Let $Q_\ell(t)\in\mathds{N}$ denote the $\ell$th queue length and define the vector $\bq(t) :=[Q_1(t), Q_2(t), \cdots, Q_L(t)]$. The policy for joint queue-and-beamforming control, referred to hereafter simply as the \emph{control policy},  is defined as a function $\pi: \mathds{N}^L\rightarrow \mathcal{V}$, and thus $\bm(t) = \pi(\bq(t))$.  A stationary policy refers to one for which $\Pr(\bm = \bv)$ for any $\mathcal{V}$ is time invariant. The control policy is further discussed in Section~\ref{Section:Policy}.

\section{Stability for Perfect CSI Feedback}\label{Section:PerfCSI}
In this section, assuming perfect CSI feedback, we analyze the stability region for zero-forcing SDMA and discuss joint queue-and-beamforming control for stabilizing queues. Consider equal power allocation over scheduled queues. The stability region is shown to be a convex polytope with derived vertices. For any arrival-rate vector in this region, queue stability is shown to be achieved by the control policy that maximizes the departure-rate-weighted sum of queue lengths. Finally, we discuss the stability region for unequal power allocation.

\subsection{Stability Region}\label{Section:StabilityRegion}
The departure rates for the SDMA system are derived. Then its stability region is defined and shown to be a convex polytope.

Several useful definitions are provided. Let $\boldsymbol{\mu}(\bm)$ denote the $L\times 1$ departure-rate vector conditioned on the control decision  $\bm$ (see Section~\ref{Section:Sys:Control}). The $\ell$th component of $\boldsymbol{\mu}(\bm)$ gives the departure rate of the $\ell$th queue and is defined as $[\boldsymbol{\mu}(\bm)]_\ell := \Pr(\SNR_\ell \geq \theta \mid \bm)$ where $\SNR_\ell$ is in \eqref{Eq:SINR:Def}. Let $\bv_\ell$ denote the $\ell$th of the $2^L$ element in the control decision space $\mathcal{V}$. A stationary policy $\pi$ can be specified by a probability vector $\bp_\pi = [p_1, p_2, \cdots, p_{2^L}]$ where $p_\ell = \Pr(\bm = \bv_\ell)$  and $\|\bp_\pi\|_1 = 1$ (see Section~\ref{Section:Sys:Control}). \footnote{$\|\ba\|_1=\sum_\ell |a_\ell|$ denotes the $L_1$ norm of the vector $\ba$. } Let $\bar{\boldsymbol{\mu}}(\bp_\pi)$ denote the departure-rate vector conditioned on $\pi$. It follows that $\bar{\boldsymbol{\mu}}(\bp_\pi) = \sum_{\ell=1}^{2^L} \boldsymbol{\mu}(\bv_\ell) p_\ell$.

The conditional departure-rate vector $\bar{\boldsymbol{\mu}}(\bp_\pi)$ is derived as follows. Let $K := \|\bm\|_1$ represent the number of scheduled queues. Zero-forcing beamforming for fixed $K \leq L$ is considered in \cite{Jindal:RethinkMIMONetwork:LinearThroughput:2008} and called \emph{partial zero-forcing} beamforming. From \cite[Lemma~2]{Jindal:RethinkMIMONetwork:LinearThroughput:2008}, given $K = k$, the effective channel power for a scheduled queue, namely $|\bff^\dagger_\ell\bh_\ell|^2$ in \eqref{Eq:SINR:Def}, follows the chi-squared distribution with $(L-k+1)$ complex degrees of freedom, denoted as  $\chi^2(L-k+1)$. Therefore the departure-rates for all scheduled queues are identical and given by $d(k)$ in \eqref{Eq:DepRatePerQueue}. Note that $d(k)$ decreases with increasing $k$ and vice versa, where $k$ also specifies the spatial multiplexing gain. It follows that $\boldsymbol{\mu}(\bm)$ can be written as $\boldsymbol{\mu}(\bm) = d(\|\bm\|_1)\bm$. Using this result, $\bar{\boldsymbol{\mu}}(\bp_\pi)$ is obtained as  shown in the following lemma.
\begin{lemma} \label{Lem:DepartRate}
Given perfect CSIT and the stationary control policy $\pi$, the conditional departure-rate vector $\bar{\boldsymbol{\mu}}(\bp_\pi)$ is given as
\begin{equation}\label{Eq:DepartRate}
\bar{\boldsymbol{\mu}}(\bp_\pi) = \sum\nolimits_{\ell=1}^{2^L} p_\ell d(\|\bv_\ell\|_1)\bv_\ell
\end{equation}
where $d(k)$ given below is the departure rate for  each of total $k$ scheduled queues \footnote{$\Gamma(a):=\int_0^\infty t^{a-1}e^{-t}dt$ and $\Gamma(a, b):=\int_b^\infty t^{a-1}e^{-t}dt$ denote the gamma and the  upper incomplete gamma functions with $a$ complex degrees of freedom, respectively.}
\begin{equation}\label{Eq:DepRatePerQueue}
d(k) = \frac{\Gamma\l(L-k+1, \frac{k\theta}{P}\r)}{\Gamma(L-k+1)}.
\end{equation}
\end{lemma}

Using these results, we now define the stability region for zero-forcing SDMA. Adopting the definition in  \cite{Szpankowski:StabCondMultiQueue:1994}, the $\ell$th queue at the base station is \emph{stable} if the queue length $Q_\ell(t)$ (in packet) satisfies
\begin{equation}\label{Eq:Stability}
\lim_{t\rightarrow\infty}\Pr(Q_\ell(t)< q) = F(q) \quad\textrm{and}\quad \lim_{q\rightarrow\infty}F(q) = 1
\end{equation}
where $F(\cdot)$ denotes a CDF. Let $\lambda_\ell$ and $\bar{\mu}_\ell$ denote the arrival and departure rates for the $\ell$th queue, respectively. It follows from Loynes' theorem \cite{BaccelliBook} that the $\ell$th queue is stable if $\lambda_\ell < \bar{\mu}_\ell$, but the stability condition for the boundary point $\lambda_\ell = \bar{\mu}_\ell$ is uncertain. To simplify our discussion, we assume stability at $\lambda_\ell = \bar{\mu}_\ell$ as in \cite{Naware:StableDelayFiniteUserALOHA:2005}. Thereby we can define the stability region for the policy $\pi$ as the closure of arrival-rate vectors for which all queues are stabilized by $\pi$. To be precise, given $\pi$ and the corresponding probability vector $\bp_\pi$, the stability region $\mathcal{A}(\bp_\pi)$ is defined as
\begin{equation}\label{Eq:StabRegion:Policy}
\mathcal{A}(\bp_\pi) := \{\boldsymbol{\lambda} \in\mathds{R}_+^L\mid \boldsymbol{\lambda}  \preceq  \bar{\boldsymbol{\mu}}(\bp_\pi)\}
\end{equation}
where $\preceq $ represents component-wise inequality and $\bar{\boldsymbol{\mu}}(\bp_\pi)$ is in Lemma~\ref{Lem:DepartRate}.  The stability region $\mathcal{C}$  for SDMA is readily defined below as the union of the stability regions for all feasible policies
\begin{equation}\label{Eq:StabReg:Def}
\mathcal{C} := \bigcup_{\|\bp_\pi\|_1=1}\mathcal{A}(\bp_\pi)
\end{equation}
where $\mathcal{A}(\bp_\pi)$ is given in \eqref{Eq:StabRegion:Policy}.

Define the set $\mathcal{R} := \l\{\boldsymbol{\mu}(\bv): \bv\in\mathcal{V}\r\}=\l\{d(\|\bv\|_1)\bv: \bv\in\mathcal{V}\r\}$ that groups the departure-rate vectors for all control decisions in the space $\mathcal{V}$. The following lemma states that the stability region $\mathcal{C}$ is a convex polytope whose vertices belong to $\mathcal{R}$.
\begin{lemma}\label{Lem:StabRegion:CSI}
Given perfect CSIT, the stability region for SDMA  is $\mathcal{C} = \overline{\mathsf{co}} \l(\mathcal{R}\r)$.
\footnote{$\overline{\mathsf{co}}$ denotes the closed convex-hull operation \cite{UrrutyBook:ConvexAnalysis:04}.}
\end{lemma}
\begin{proof}
See Appendix~\ref{Lem:StabRegion:CSI}.
\end{proof}
Lemma~\ref{Lem:StabRegion:CSI} shows that $\mathcal{C}$ is a \emph{convex polytope} with a subset of points in  $\mathcal{R}$ being the vertices. Other points are not extreme points of $\mathcal{C}$. They lie either on the surface of the polytope $\mathcal{C}$ or inside it.

We are now ready to specify the vertices of the stability region $\mathcal{C}$. To this end, some useful notation is introduced.  Conditioned on $K=k$, the control decision $\bm$ belongs to the set $\mathcal{V}_k := \l\{\bv \in \mathcal{V}: \|\bv\|^2=k\r\}$. The corresponding set of departure-rate vectors is defined as $\mathcal{R}_k := \l\{d(k)\bv: \bv\in\mathcal{V}_k\r\}$ where $d(k)$ is given in \eqref{Eq:DepRatePerQueue}. Note that $\bigcup_{0\leq k\leq L}\mathcal{V}_k=\mathcal{V}$, $\bigcup_{0\leq k\leq L}\mathcal{R}_k=\mathcal{R}$, and $\mathcal{V}_0 = \mathcal{R}_0 = \{\mathbf{0}\}$. Define the index set
\begin{equation}\label{Eq:Index:Def}
\mathcal{I} := \l\{0\leq k \leq L : kd(k) > \max_{0\leq m < k}md(m) \r\}.
\end{equation}
The main result of this section is given in the following theorem.
\begin{theorem}\label{Theo:StabRegion}
For perfect CSIT, the stability region $\mathcal{C}$ is a convex polytope in $\mathds{R}^L$ whose vertices (extreme points) are given in the following set
\begin{equation}\label{Eq:ExtPts}
\mathsf{ext}\ \mathcal{C} := \bigcup\nolimits_{k \in\mathcal{I}} \mathcal{R}_k.
\end{equation}
\end{theorem}
\begin{proof}
Essentially, we prove this theorem by showing that $\mathsf{ext}\ \mathcal{C}$ contains \emph{all} the \emph{exposed} points of $\mathcal{C}$. \footnote{A point $\bg\in\mathcal{C}$ is \emph{exposed} if there exists a hyper-plane $\mathcal{P}$ supporting $\mathcal{C}$ such that $\mathcal{P}\cap\mathcal{C} = \bg$ \cite{UrrutyBook:ConvexAnalysis:04}. } Details are presented in Appendix~\ref{App:StabRegion}.
\end{proof}
Note that $\mathsf{ext}\ \mathcal{C}$ contains the origin $\mathbf{0}$. Theorem~\ref{Theo:StabRegion} suggests the following simple procedure for constructing $\mathcal{C}$, where $R_{\max}$ denotes the maximum downlink packet rate.
\begin{enumerate}
\item Initialization: $k = 0$, $\mathsf{ext}\ \mathcal{C} = \{\mathbf{0}\}$, and $R_{\max} = 0$;
\item Let $k = k + 1$. If $kd(k) > R_{\max}$, $R_{\max} = kd(k)$ and $\mathsf{ext}\ \mathcal{C} = \mathsf{ext}\ \mathcal{C}\cup\{\mathcal{R}_k\}$;
\item If $k \leq L$, repeat 2); otherwise go to $4)$
\item Construct a polytope $\mathcal{C}$ using  the points in $\mathsf{ext}\ \mathcal{C}$ as vertices.
\end{enumerate}
Step $2)$ implies that scheduling $k$ queues enlarges the stability region only if it increases the sum-packet rate with respect to scheduling fewer queues. Note that the number of scheduled queues achieving $R_{\max}$ varies with the SNR and need not be equal to the maximum $L$. In other words, the shape of the stability region $\mathcal{C}$ varies with SNR as illustrated by the following example.

\emph{Example $1$:} The stability region for $L=3$ is plotted in Fig.~\ref{Fig:ConvexHull}(a) for $P/\theta = 0.5$ and Fig.~\ref{Fig:ConvexHull}(b) for $P/\theta = 10$, corresponding to a low and a high SNRs, respectively.  As observed from Fig.~\ref{Fig:ConvexHull}, $\mathsf{ext}\ \mathcal{C}$ reduces to $\mathsf{ext}\ \mathcal{C} = \mathcal{R}_1\cup\{\mathbf{0}\}$ for the low SNR  and has the maximum size, namely $\mathsf{ext}\ \mathcal{C} = \mathcal{R}$, at the high SNR.  This observation is formally stated in the following proposition.

\begin{figure}
\centering
  \subfigure[Low SNR $\l(\frac{P}{\theta} = 0.5\r)$]{\includegraphics[width=6.5cm]{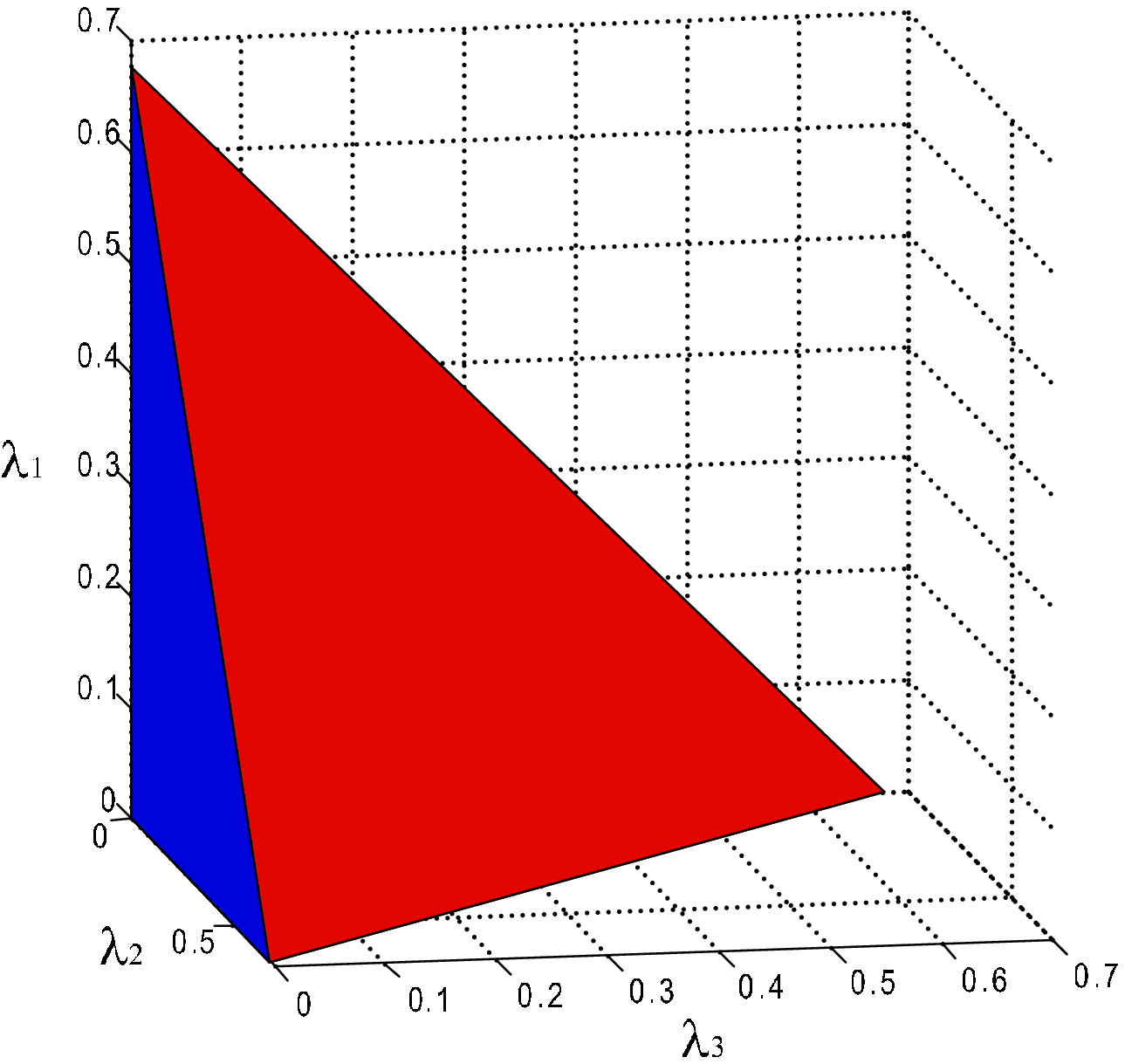}}\hspace{30pt}
    \subfigure[High SNR $\l(\frac{P}{\theta} = 10\r)$]{\includegraphics[width=6.5cm]{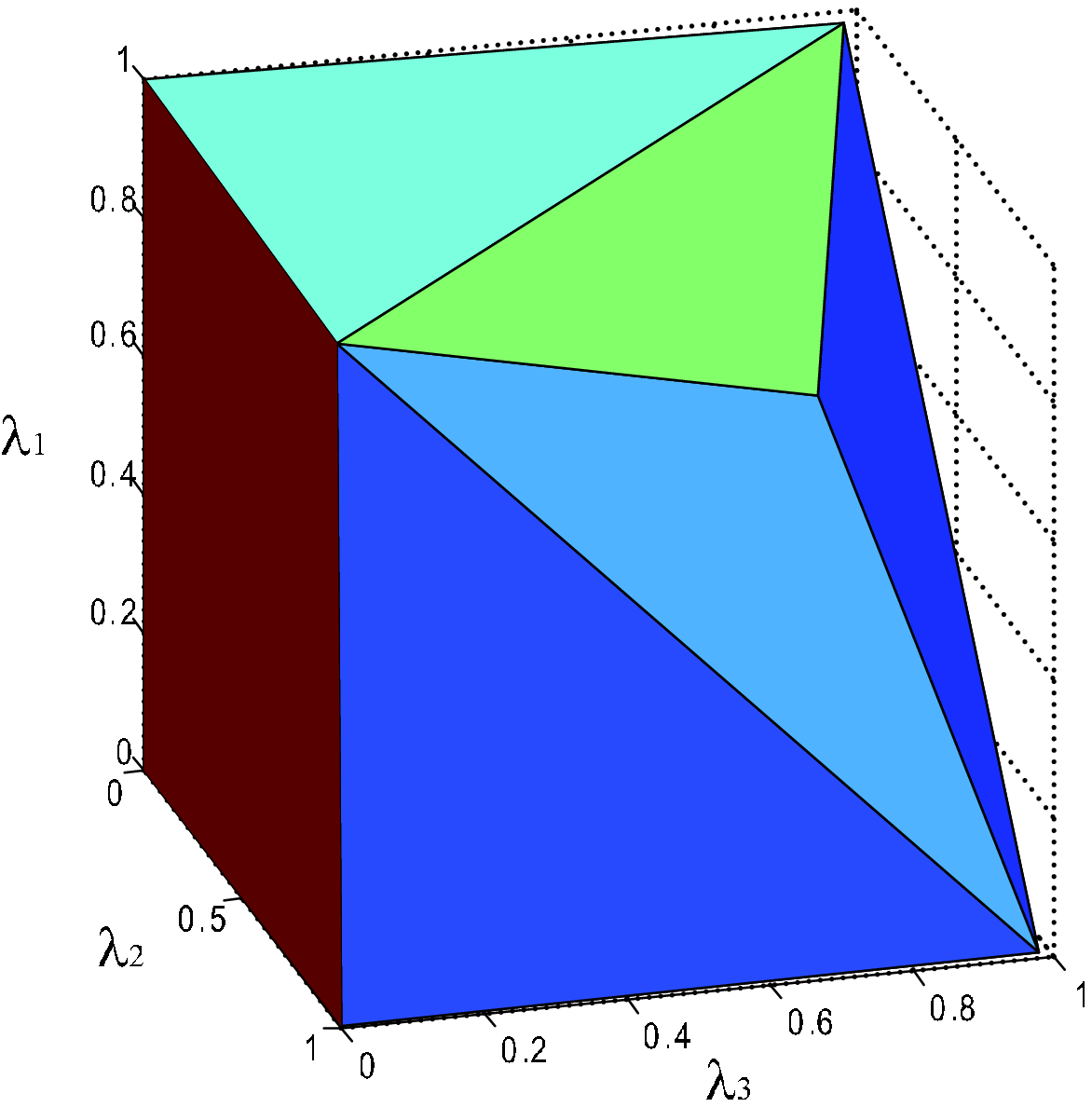}}
  \caption{Stability region for $L=3$ and (a) low SNR: $P/\theta = 0.5$ or (b) high SNR: $P/\theta = 10$}\label{Fig:ConvexHull}
\end{figure}

\begin{proposition}\label{Prop:StabExtreme} There exist $\tau> 0$ such that $\mathsf{ext}\ \mathcal{C} = \mathcal{R}_1\cup\{\mathbf{0}\} \ \forall \ P \leq \tau$, and $P_0 > 0$ such that $\mathsf{ext}\ \mathcal{C} = \mathcal{R} \ \forall \ P \geq P_0$.
\end{proposition}
\begin{proof} See Appendix~\ref{App:StabExtreme}.
\end{proof}

The above proposition suggests that at low SNRs the control policy based on time-division-multiple-access (TDMA) is \emph{stability optimal}, namely  stabling queues for any arrival-rate vector in the stability region; at high SNRs, the stability-optimal policy should schedule multiple queues for SDMA. Discussion on the control policy is presented in the next section.

\subsection{Joint Queue and Beamforming Control}\label{Section:Policy}
In this section, we discuss stability optimal policies for joint queue-and-beamforming control (see Section~\ref{Section:Sys}). 

If the arrival rates of the queues are known, the queue stability can be achieved by pre-defined time-sharing between scheduling different subsets of queues. An arrival rate vector $\boldsymbol{\lambda}\in\mathcal{C}$ can be written as a convex combination of the vertices of $\mathcal{C}$ since it is a convex polytope. Specifically, $\boldsymbol{\lambda} = \sum_{n=1}^{T} a_n\bu_n$, where $\bu_n$ is the $n$th element of $\mathsf{ext}\ \mathcal{C}$ in \eqref{Eq:ExtPts}, $T :=|\mathsf{ext}\ \mathcal{C}|$,\footnote{$|\mathcal{X}|$ represents the cardinality of the set $\mathcal{X}$.} and $\{a_n\}$ are positive and satisfy the constraint  $\sum_{n=0}^{T}a_n = 1$. We can find at least one point $\boldsymbol{\lambda}'$ on the boundary of $\mathcal{C}$ such that $\boldsymbol{\lambda}'\succeq\boldsymbol{\lambda}$. Write $\boldsymbol{\lambda}' = \sum_{n=1}^{T} b_n\bu_n$ with  $b_n \geq 0\ \forall \ n$ and $\sum_n^Tb_n=1$. Recall from earlier discussion that we can write  $\bu_n  =d(\|\bv_n\|_1)\bv_n$ where $\bv_n\in\{0,1\}^L$. Therefore, given $\boldsymbol{\lambda}$, queues  can be stabilized by any time-sharing policy where the  fraction of time when $\bm(t) = \bv_n$ is $b_n$. Alternatively, a randomized policy can be applied such that $\Pr(\bm = \bv_n) = b_n$ \cite{GeorgNeelyBook}.

In practice, the packet arrival rates are usually unknown. For this case, the queues can be stabilized using a policy that controls queue and beamforming based on queue lengths rather than packet arrival rates. Using the principle of \emph{maximum weight matching} \cite{GeorgNeelyBook}, the following proposed control policy maximizes the departure-rate-weighted sum of queue lengths
\begin{equation}\label{Eq:QueControl}
\pi: \bm(t) = \arg\max_{\bv\in\mathcal{V}} d(\|\bv\|_1)\bv^T\bq(t).
\end{equation}
The following proposition states the optimality of the above policy.

\begin{proposition}\label{Prop:Lyapunov} The control policy in \eqref{Eq:QueControl} stabilizes all $L$ queues if the arrival rate vector lies in the stability region $\mathcal{C}$ in Theorem~\ref{Theo:StabRegion}.
\end{proposition}

Proving Proposition~\ref{Prop:Lyapunov} uses the technique of \emph{Lyapunov drift} \cite{GeorgNeelyBook}  and follows a similar procedure as that for \cite[Theorem~3]{Neely:PowrAllocationRoutMultibeamSatellites:2003}. The  proof is omitted for brevity.

\subsection{Stability Region with Power Control}\label{Section:PowerControl}
Equal power allocation over scheduled queues is assumed in the preceding sections. In this section, we relax this assumption and derive the stability region with power control. Compared with equal power allocation, power control provides an additional degree of freedom for enlarging the stability region as illustrated in Fig.~\ref{Fig:PowerControl}. Some notation is introduced as follows, where the slot index $t$ is omitted for brevity. Denote  the power vector as $\bp = [P_1, P_2, \cdots, P_L]^T$ where $\sum_\ell P_\ell = P$ and $P_\ell$ gives the transmission power allocated to the $\ell$th queue. By extending Lemma~\ref{Lem:DepartRate}, the departure rate for the scheduled queue with the index $\ell$ is  $d(P_\ell, \bm) := \Gamma\l(L-\|\bm\|_1 +1, \frac{\theta}{P_\ell}\r)\times [\bm]_\ell$. We define the departure rate vector $\boldsymbol{\mu}(\bp, \bv)$ with the $\ell$ component being $d(P_\ell, \bm)$. The stability region for SDMA with power control is shown in the following proposition. The proof follows a similar procedure as that for Lemma~\ref{Lem:StabRegion:CSI} and is thus omitted.

\begin{proposition}\label{Prop:StabRegion:Power} The stability region  for a zero-forcing SDMA system with power control is \footnote{For the vector $\bv$, $\diag(\bv)$ denotes a matrix with diagonal elements taken from $\bv$ and other elements being zeros.  }
\begin{equation}
\mathcal{C}_P = \textsf{co}\l\{\bigcup_{\bv\in\mathcal{V}}\bigcup_{\bp^T\bv\leq P}\l\{\boldsymbol{\mu}(\bp, \bv)\diag(\bv)\r\}\r\}.
\end{equation}
\end{proposition}
Unlike $\mathcal{C}$, $\mathcal{C}_P$ is not a convex polytope. Moreover, $\mathcal{C}\in \mathcal{C}_P$. The stability-optimal control policy in \eqref{Eq:QueControl} can be easily extended to include power control.

\begin{figure}
\centering
  \includegraphics[width=8cm]{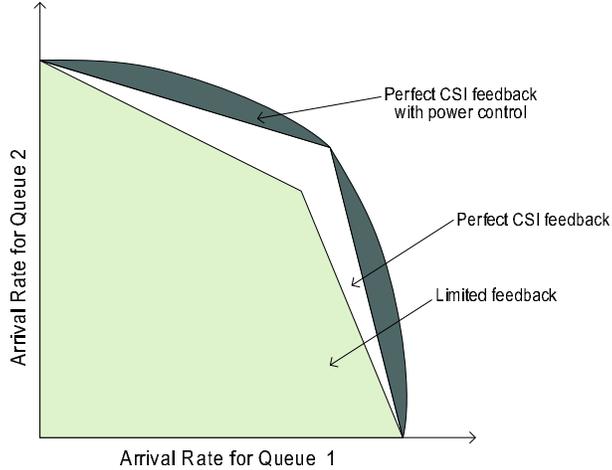}\\
  \caption{Effects of power control and limited feedback on the stability region for $L=2$}\label{Fig:PowerControl}
\end{figure}

\section{CSI Feedback Requirements}\label{Sectoin:FbRequire}
The stability analysis in the preceding section assumes perfect CSIT. In this section, we consider limited feedback and
analyze the feedback overhead required for constraining the performance degradation due to imperfect CSIT. The required number of feedback bits per mobile $B$ is obtained for achieving the stability region close to the perfect-CSI counterpart. We also derive $B$ required for approaching optimal delay performance.

\subsection{Stability Region for SDMA with Limited Feedback}
CSI quantization errors not only reduce received signal power but also cause mutual interference between packets transmitted using SDMA. Consequently, the departure rate of each queue decreases, shrinking the stability region as illustrated in Fig.~\ref{Fig:PowerControl}. In this section, we derive $B$ for bounding the difference in stability region between the cases of perfect CSIT and limited feedback. 

The power distributions of the received signals and multiuser interference are given in the following two lemmas. Lemma~\ref{Lem:QBeam:Sig} provides the distribution of the signal term $|\bff^\dagger_\ell\bh_\ell|^2$ of the SINR in \eqref{Eq:SINR:Def}. 

\begin{lemma}\label{Lem:QBeam:Sig} The random variable $|\bff^\dagger_\ell\bh_\ell|^2$ follows the exponential distribution with unit mean for the number of scheduled queues $K=L$. For $K < L$,  $|\bff^\dagger_\ell\bh_\ell|^2:=(\epsilon_\ell Z - \rho_\ell Q)$, where $\epsilon_\ell$ is the CSI error, $\epsilon_\ell \leq \rho_\ell\leq \epsilon_\ell$, and the random variables $Z$ and $Q$ follow independent $\chi^2(L-K+1)$ and $\chi^2(K-1)$ distributions, respectively,
\end{lemma}
\begin{proof}See Appendix~\ref{App:QBeam:Sig}
\end{proof}
This lemma generalizes \cite[Lemma~2]{Jindal:RethinkMIMONetwork:LinearThroughput:2008} for perfect CSIT to include the case of limited feedback. In other words, substituting $\epsilon_\ell=0$ into Lemma~\ref{Lem:QBeam:Sig} gives \cite[Lemma~2]{Jindal:RethinkMIMONetwork:LinearThroughput:2008}. 

The results in Lemma~\ref{Lem:QBeam:Sig} are explained as follows. For full spatial multiplexing $K=L$, all spatial degrees of freedom contribute to suppressing interference between queues. Consequently, the beamforming vector for each queue is chosen independent of the corresponding vector channel \cite{Jindal:MIMOBroadcastFiniteRateFeedback:06}. That is to say, CSI is not used for strengthening received signal power and thus it is independent of the CSI quantization error as observed from Lemma~\ref{Lem:QBeam:Sig}. For $K<L$, a fraction of the spatial degrees of freedom ($L-K+1$) provide diversity gains to the data links, which are realized by using feedback CSI. Thus the CSI error causes the received power loss, which is quantified by the coefficients $(1-\epsilon_\ell)$ and $\rho_\ell$ in Lemma~\ref{Lem:QBeam:Sig}. 

The distribution of the interference term $|\bff^\dagger_n\bh_m|^2$ of the SINR in \eqref{Eq:SINR:Def} is given in Lemma~\ref{Lem:Interf} that follows from \cite[Lemma~2]{Jindal:MIMOBroadcastFiniteRateFeedback:06}.

\begin{lemma}[\cite{Jindal:MIMOBroadcastFiniteRateFeedback:06}]\label{Lem:Interf}
The random variable $T_{n,m}:=|\bff^\dagger_n\bh_m|^2$ for $n\neq m$ is the product of the $\chi^2(L)$ random variable $(Z+Q)$ and an independent $\beta(1, L-2)$ random variable,\footnote{A $\beta(c, a/b)$ random variable $T$ satisfies $0\leq T\leq 1$ and has the CDF $\Pr(T\leq t) = b\int_0^t z^{a-1}(1-z^b)^{c-1}dz$ with $a>0$, $b>0$ and $z>0$ \cite{Jindal:MIMOBroadcastFiniteRateFeedback:06}. } where $Z$ and $Q$ are identical to those in Lemma~\ref{Lem:QBeam:Sig}.
\end{lemma}

As observed from Lemma~\ref{Lem:Interf}, the number of scheduled queues (or spatial streams) $K$ has no effect on the distribution of each interference power component $\gamma\epsilon |\bff^\dagger_n\bh_m|^2$, but determines the number of such components (see \eqref{Eq:SINR:Def}).

Using Lemma~\ref{Lem:Interf} and Lemma~\ref{Lem:Interf}, we obtain the loss on the departure rate for scheduled queues due to limited feedback. Define the departure rate conditioned on $K=k$ as $\hat{d}(k):= \Pr\l(\SINR_\ell \geq \theta \mid K=k\r)$ where $\SINR$ is given in \eqref{Eq:SINR:Def}. The loss on  the departure rate can be characterized by the ratio $\hat{d}(k)/d(k)$. This ratio should be kept close to one by providing sufficiently accurate CSI feedback. Thus we define $\delta_k := 1-\hat{d}(k)/d(k)$ and
\begin{equation}\label{Eq:LossFact}
\delta := \max_k\delta_k = 1 - \min_{1\leq k\leq L}\hat{d}(k)/d(k)
\end{equation}
which characterizes the maximum loss on the departure rate due to limited feedback. Given $\delta$, the required number of CSI feedback bits per mobile is shown in the following lemma.

\begin{lemma}\label{Lem:Loss:Bnds}
Given $\delta$, it is sufficient for each mobile to provide the following number of CSI feedback back bits $B$
\begin{equation}\label{Eq:Fb:bits}
B(\delta) = -(L-1)\log_2\delta + \kappa
\end{equation}
where $\kappa := (L-1)\log_2\l(L(1+L\theta)\l(1+\frac{\theta}{P}\r)\r)$.
\end{lemma}
\begin{proof}See Appendix~\ref{App:Loss:Bnds}.
\end{proof}

Let $\hat{\mathcal{C}}$ represent the stability region for zero-forcing SDMA with limited feedback. Using Lemma~\ref{Lem:Loss:Bnds}, the main result of this section is obtained and shown in the following theorem.

\begin{theorem}\label{Theo:StabRegion:LimCSI} Given $0<\delta<1$, $B(\delta)$ in \eqref{Eq:Fb:bits} is sufficient for bounding the stability region $\hat{\mathcal{C}}$ as $(1-\delta)\mathcal{C} \subset \hat{\mathcal{C}} \subset \mathcal{C}$.
\end{theorem}

As observed from \eqref{Eq:Fb:bits}, $B$ increases rather slowly (logarithmically) with the inverse of $\delta$. For small values of $\delta$, $B\approx (L-1)\log_2\frac{1}{\delta}$, where $B$ increases linearly with $(L-1)$. In addition, $B$ increases with the ratio $\theta/P$. Note that a large ratio corresponds to high coding rate or low transmission power and vice versa.

\subsection{Queueing Delay for Limited Feedback}
As mentioned, limited feedback reduces the departure rate for each queue and thereby increases queueing delay due to limited feedback.  In this section, we derive $B$ for  bounding the degradation of the system delay performance. We consider both the Poisson and general packet arrival processes.

\subsubsection{Poisson Arrival}
Without loss of generality, we consider the queue with the index $1$ and the Poisson arrival rate $\lambda$.  Given a control policy $\pi$ specified by the probability vector $\bp_\pi$, the average departure rate for perfect CSIT is defined as $\mu := [\boldsymbol{\bar{\mu}}(\bp_\pi)]_1$ where $\boldsymbol{\bar{\mu}}(\bp_\pi)$ is given in Lemma~\ref{Lem:DepartRate}. Moreover, let $\hat{\mu}$ denote the counterpart of $\mu$ for limited feedback.
For Poisson arrivals, the average delay $\widehat{W}$ follows the Pollaczek-Khinchin formula  \cite{BertsekasBook:DataNetwk:92}
\begin{equation}\label{Eq:P-K}
\widehat{W} = \frac{\lambda \E[X^2]}{2(1-\lambda/\hat{\mu})}
\end{equation}
where $X$ denotes the service time. The second moment $\E[X^2]$ for ARQ transmission is given by $\E[X^2]=  (2-\hat{\mu})/\hat{\mu}^2$ \cite{BertsekasBook:DataNetwk:92}. Thus it follows from \eqref{Eq:P-K} that
\begin{equation}\label{Eq:P-K:a}
\widehat{W} = \frac{\lambda(2-\hat{\mu})}{2\hat{\mu}(\hat{\mu}-\lambda)}.
\end{equation}
Let $W$ denote the average delay for perfect CSIT. To limit the effect of CSI inaccuracy, it is desirable to bound the ratio $\widehat{W}/W\leq M$ where $M>1$. Given $M$, the required number of feedback bits $B$ is shown in the following proposition.

\begin{proposition}\label{Prop:Delay:Poisson}
Consider Poisson packet-arrival processes. Define $\tau := 1-\frac{\lambda}{\mu}$. To satisfy the constraint $W/\widehat{W}\leq M$ with $M>1$, it is sufficient to provide $B(\delta^+)$ CSI feedback bits per mobile, where $B(\cdot)$ follows Lemma~\ref{Lem:Loss:Bnds} and $\delta^+$ is given as
\begin{equation}
\delta^+ := \frac{1}{2}\l[1-\sqrt{1-\frac{4\l(1-\frac{1}{M}\r)\tau}{(1+\tau)^2}}\r].
\end{equation}
\end{proposition}
\begin{proof}See Appendix~\ref{App:Delay:Poisson}.
\end{proof}
The dependence of $M$ on $B$ is illustrated in the following example.

\emph{Example $2$:} In Fig.~\ref{Fig:DelayPoisson}, the  $M$ versus $B$ curves are plotted, where $B$ is computed from $M$ using Proposition~\ref{Prop:Delay:Poisson} for $L = 3$, $\theta = 3$ and $P= 12$ dB. Note that fractional values for $B$ in Fig.~\ref{Fig:DelayPoisson} should be rounded to integers for practical applications. As observed from Fig.~\ref{Fig:DelayPoisson}, as $B$ decreases, $M$ increases very rapidly because $\hat{\mu}$ approaches $\lambda$. In the regime of large $B$, $M$ converges gradually to $1$ with increasing $B$.

\begin{figure}
\centering
  \includegraphics[width=9cm]{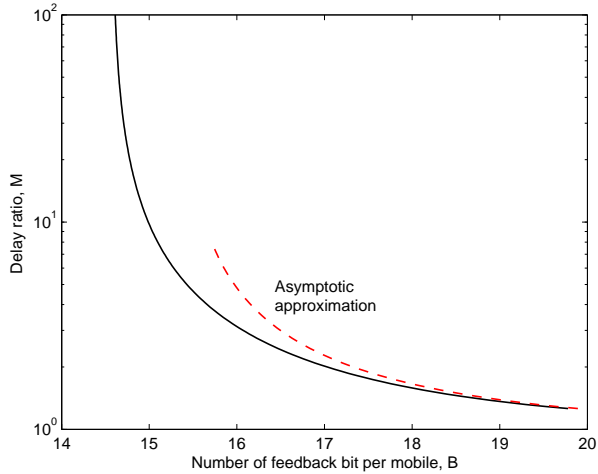}\\
  \caption{The delay ratio $M$ versus the number of feedback bits per mobile $B$ (solid line) in Proposition~\ref{Prop:Delay:Poisson} for Poisson arrivals, $L = 3$, $\theta = 3$ and $P= 12$ dB. The asymptotic approximation for $M\rightarrow 1$ is plotted in dashed line.}\label{Fig:DelayPoisson}
\end{figure}

Finally, for $B\rightarrow\infty$ and thus $M\rightarrow 1$, the relation between $M$ and $B$ is simplified in the following corollary of Proposition~\ref{Prop:Delay:Poisson}.
\begin{corollary}\label{Cor:DelayPoisson}
For $M\rightarrow 1$, it is sufficient to scale $B$ with $M$ as
\begin{equation}\label{Eq:B:Asymp}
B = (L-1)\log_2\frac{M}{M-1} + (L-1)\log_2\frac{(1+\tau)^2}{\tau} + \kappa + O\l(\l(1-1/M\r)^2\r)
\end{equation}
where $\kappa$ is given in Lemma~\ref{Lem:Loss:Bnds}.
\end{corollary}

\emph{Example $2$ (continue):}  In Fig.~\ref{Fig:DelayPoisson}, the curve of $M$ versus the asymptotic approximation $B = (L-1)\log_2\frac{M}{M-1} + \kappa$ from \eqref{Eq:B:Asymp} is plotted in dashed line. As observed from Fig.~\ref{Fig:DelayPoisson}, the approximation is accurate for $B\geq 15$ bit, corresponding to $M\leq 1.5$.

\subsubsection{General Arrival} In this section, we consider the same settings as in the preceding section but with general arrival processes. Let $\widehat{V}$ represent the delay for a particular packet and thus $\hat{W} = \E[\hat{V}]$. The effect of limited feedback on the CCDF of $\widehat{V}$ is analyzed.  Let $X$ and $\hat{Y}$ denote the inter-arrival time and the service time, respectively. Using Kingman's bound for G/G/1 queues, the CCDF of $\widehat{V}$ is bounded as \cite{GallagerBook:StochasticProcs:95}
\begin{equation}\label{Eq:Kingman}
\Pr(\widehat{V}\geq t) \leq \exp(-\hat{r}^\star t)
\end{equation}
where $\hat{r}^\star$ is the solution of the following equation
\begin{equation}
\E[\exp(r(\hat{Y}-X))] = 1. \label{Eq:ExpEq}
\end{equation}
As observed from \eqref{Eq:Kingman}, the probability that $\widehat{V}$ exceeds  the threshold $t$ decreases exponentially with $t$, and $\hat{r}^\star$ determines the decrease rate.  Let $r^\star$ represent the counterpart of $\hat{r}^\star$ for perfect CSIT. Define the lose factor $\sigma := 1- \hat{\mu}/\mu$. Note that $\sigma\leq \delta$ with $\delta$ in \eqref{Eq:LossFact}. The relation between $\hat{r}^\star$ and $r^\star$ for a small $\delta$ is established in the following proposition. Its proof uses perturbation theory \cite{SimBook:FirstLookPerturbationTheory:97}.
\begin{lemma}\label{Lem:Delay:General}For $\sigma\rightarrow 0$, the CCDF exponential factors $\hat{r}^\star$ and $r^\star$ for respectively limited feedback and perfect CSIT is related as
\begin{equation}
\hat{r}^\star = r^\star - f(r^\star)\sigma +O(\sigma^2)
\end{equation}
where
\begin{equation}\label{Eq:Beta}
f(r^\star) := \frac{1-e^{-r^\star}}{\mu\E[e^{-r^\star X}X]-e^{-r^\star}}.
\end{equation}
\end{lemma}
\begin{proof} See Appendix~\ref{App:Delay:General}.
\end{proof}

For high-resolution CSI feedback ($B\rightarrow\infty$), we show that Kingman's bounds  for perfect CSIT and limited feedback are related by linear scaling with a factor converging to  $1$ as $B\rightarrow \infty$. This result is given in the following proposition.
\begin{proposition}\label{Prop:Delay:General}
For limited feedback with
\begin{equation}\label{Eq:B:Asymp:a}
B = -(L-1)\log_2\eta + \kappa + O(\eta), \quad \eta \rightarrow 0
\end{equation}
where $\kappa$ is given in Lemma~\ref{Lem:Loss:Bnds}, the delay CCDF is bounded as
\begin{equation}\label{Eq:CCDF:Bnd}
\Pr(\widehat{V}\geq t) \leq e^{-r^\star t}(1+\eta).
\end{equation}
\end{proposition}
Proposition~\ref{Prop:Delay:General} is proved by combining \eqref{Eq:Kingman} and Lemma~\ref{Lem:Delay:General} and setting $\eta := f(r^\star) \sigma t $. By integrating both sides of \eqref{Eq:CCDF:Bnd}, the average delay $\widehat{W} = \E[\widehat{V}]$ can be bounded as
\begin{equation}
\widehat{W} \leq  \frac{1+\eta}{r^\star},\quad \eta \rightarrow 0.
\end{equation}
Similarly, $W \leq  1/r^\star$. Therefore, for $B$ given in \eqref{Eq:B:Asymp:a}, CSI inaccuracy increases the delay upper-bound by the factor $(1+\eta)$.

\section{Numerical Results}\label{Section:Simualtion}
In this section, we present numerical results and further discuss the effects of limited feedback on the stability region and queueing delay of the SDMA system.


Fig.~\ref{Fig:LossFact} displays $(1-\delta)$ with $\delta$ in \eqref{Eq:LossFact} for different combinations of transmission power $P$ and the number of transmit antennas $L$. The number of feedback bits per mobile is fixed at $B=12$ bit and the SINR threshold $\theta = 3$. It is observed from Fig.~\ref{Fig:LossFact} that $(1-\delta)$ is insensitive to the changes on $P$ for $P\geq 5$ dB. This is reflected in the expression of $B$  in \eqref{Eq:Fb:bits} where $B$ is independent of $P$ for $P\rightarrow \infty$. Next, from Fig.~\ref{Fig:LossFact}, $(1-\delta)$ is found to decrease rapidly as $L$ increases. To keep $(1-\delta)$ constant, $B$  should increase linearly with $(L-1)$ as shown in \eqref{Eq:Fb:bits}. Thus $L$ is an important
factor for determining the CSI feedback requirement. It is worthy to mention that the same feedback requirement is found in \cite{Jindal:MIMOBroadcastFiniteRateFeedback:06} for the SDMA system with the ergodic capacity as the performance metric.

\begin{figure}
\centering
  \includegraphics[width=10cm]{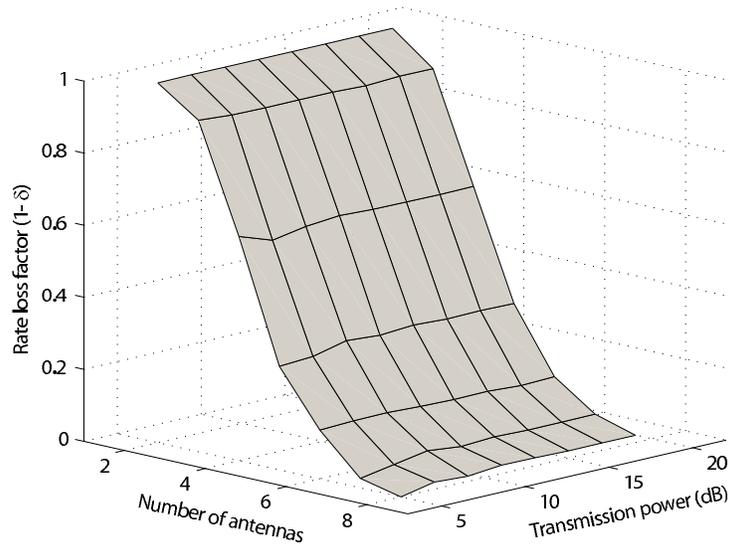}\\
  \caption{The loss factor $(1-\delta)$ for the stability region for different combinations of transmission power and number of antennas. The number of feedback bits per mobile $B=12$ bit and the SINR threshold $\theta = 3$. }\label{Fig:LossFact}
\end{figure}

Fig.~\ref{Fig:DepartRate} shows the departure rate per queue $d(k)$ for different numbers of scheduled queues $K$ and $L=\{2, 3, 4, 5\}$. The solid and dashed curves correspond to perfect CSIT and limited CSI feedback, respectively. We set  $(1-\delta) = 90\%$ and   $B$ by evaluating \eqref{Eq:Fb:bits}. The vertical bars in Fig.~\ref{Fig:DepartRate} specify $10\%$ rate loss with respect to the departure rate for perfect CSIT. It is observed from Fig.~\ref{Fig:DepartRate} that the rate loss due to limited feedback is within $\delta = 10\%$. This confirms the sufficiency of $B$ in \eqref{Eq:Fb:bits}. Moreover, Fig.~\ref{Fig:DepartRate} shows that for fixed $L$, the rate loss due to limited feedback is maximized at $L$ scheduled queues.

\begin{figure}
\centering
  \includegraphics[width=9cm]{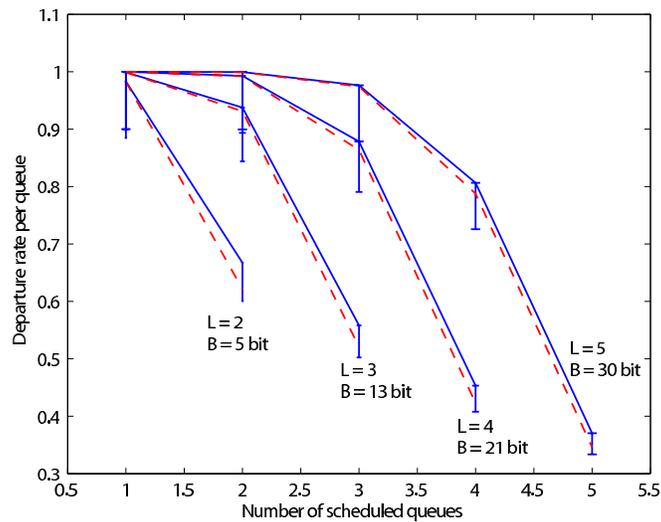}\\
  \caption{Departure rate per queue for different numbers of scheduled queues. The solid and dashed curves correspond to perfect and limited CSI feedback, respectively. The vertical bars specify $10\%$ rate loss with respect to the case of perfect CSIT. The number of antennas $L=\{2, 3, 4, 5\}$ and the SINR threshold $\theta = 3$. }\label{Fig:DepartRate}
\end{figure}

Fig.~\ref{Fig:QueueLen} displays the curves of average queue length versus arrival rate per queue. The simulation parameters are $L=4$, $K=4$ and $P=12$ dB. The curves for limited feedback are plotted in solid line and that for perfect CSIT in dashed line. For limited feedback, $B = \{8, 10, 12, 20\}$ bit and the SINR threshold $\theta = 3$. For given average queue length, the gain in arrival rate  for additional feedback bits is more significant for smaller $B$. For example, for the average queue length of $50$ packets, the gains are  $0.6$ and $0.4$ for increasing $B$ from $8$ to $10$ bit and $10$ to $12$ bit, respectively.  With $B=20$ bit, the average queue length (or average delay)  for limited feedback is close to that for that for perfect CSIT.

\begin{figure}
\centering
  \includegraphics[width=10cm]{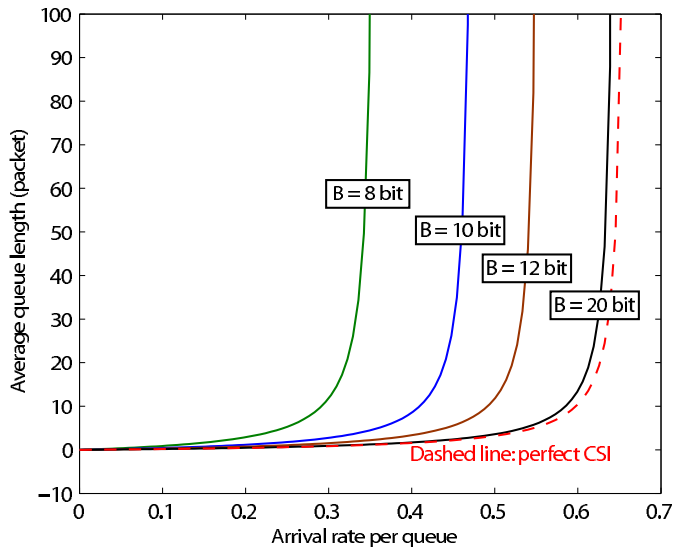}\\
  \caption{Average queue length versus arrival rate per queue  for both limited and perfect CSI feedback. The numbers of transmit antennas and scheduled queue are both equal to $4$; the transmission power is $P=12$ dB; for limited feedback, the number of feedback bit per mobile is $B = \{8, 10, 12, 20\}$; the SINR threshold is $\theta = 3$.}\label{Fig:QueueLen}
\end{figure}

\section{Conclusion}\label{Section:Conlusion}
In this paper, we consider the SDMA system with a $L$-antenna base station and $L$ single-antenna mobiles. Given perfect CSIT and equal power allocation over scheduled queues, we have proved that the stability region of the SDMA system is a convex polytope.  For any arrival-rate vector in this region, queues are stabilized by the control policy that maximizes the departure-rate-weighted sum of queue lengths. In addition, the stability region with power control has been derived. The stability region for case of limited feedback has been shown to be equal to the perfect-CSIT counterpart scaled by a factor that is determined by the number of CSI feedback bits per mobile $B$. We have also investigated the effects of limited feedback on the system delay performance. For Poisson arrival processes, limited feedback increases the average delay by a linear factor $M$. For $M\rightarrow 1$, the required $B$ has been shown to be $B = O\(-\log_2(1-1/M)\r)$. For general arrival processes, the tail probability of the delay for a particular packet can be bounded by Kingman's bound for perfect CSIT scaled by a factor $(1+\eta)$. For $\eta\rightarrow 0$, the required $B$  has been proved to be $B=O(-\log_2\eta)$. The above results constitute a set of guidelines for designing practical SDMA systems with limited feedback and for delay-sensitive applications.

This work also opens several issues for future research. First, the joint queue-and-beamforming control approach can be extended to SDMA systems with a large number of mobiles. Designing these systems should account for the effects of  multiuser diversity gains on stability, delay and CSI overhead along the lines of \cite{SwannackWorenell:LowComplexMuScheduleMIMOBC:2004, SharifHassibi:CapMIMOBroadcastPartSideInfo:Feb:05, YooJindal:FiniteRateBroadcastMUDiv:2007, SharifHassibi:DelayConsiderOppScheduleBC:2007}. Second, the simplified CSI quantization model can be replaced by the actual codebook-based quantization \cite{Love:OverviewLimitFbWirelssComm:2008}. Thereby the stability and delay related criteria can be included in codebook optimization.  Finally, alternative beamforming algorithms such as minimum-mean-square-error (MMSE) or nonlinear precoding can be studied in the current framework.

\appendix

\subsection{Proof for Lemma~\ref{Lem:StabRegion:CSI}}\label{App:StabRegion:CSI}
From Lemma~\ref{Lem:DepartRate}, \eqref{Eq:StabRegion:Policy} and \eqref{Eq:StabReg:Def}, the stability region $\mathcal{C}$ can be written as
\begin{equation}\label{App:StabReg:a}
\mathcal{C} = \bigcup_{\sum_\ell p_\ell=1}\bigcup_{0\leq a \leq 1} \l\{a\sum\nolimits_{\ell=1}^{2^L}p_\ell d(\bv_\ell)\bv_\ell\r\}.
\end{equation}
Let $\br_\ell$ denote the $\ell$th element of $\mathcal{R}$. Without loss of generality, let $\br_1$ be the all-zero element $\mathbf{0}$ of $\mathcal{R}$. By the definition of $\mathcal{R}$, \eqref{App:StabReg:a} can be rewritten as
\begin{eqnarray}
\mathcal{C} &=& \bigcup_{\sum_\ell p_\ell=1}\bigcup_{0\leq a \leq 1} \l\{a\sum\nolimits_{\ell=2}^{2^L}p_\ell \br_\ell + ap_1 \mathbf{0}\r\}\nn\\
&=& \bigcup_{\sum_\ell p_\ell=1}\bigcup_{0\leq a \leq 1} \l\{a\sum\nolimits_{\ell=2}^{2^L}p_\ell \br_\ell + \l(1-a\sum\nolimits_{\ell=2}^{2^L}p_\ell \r)\mathbf{0}\r\}.\label{App:StabReg:b}
\end{eqnarray}
Define $p'_1 = 1-a\sum_{\ell = 2}^{2^L} p_\ell$ and $p'_\ell = ap_\ell$ for $2\leq \ell \leq 2^L$. From \eqref{App:StabReg:b}
\begin{eqnarray}
\mathcal{C} &=& \bigcup_{\sum_\ell p'_\ell=1}\l\{\sum\nolimits_{\ell=1}^{2^L}p'_\ell \br_\ell \r\}.\label{App:StabReg:c}
\end{eqnarray}
Since $p'_\ell\geq 0 \ \forall \ell$ and $\sum_\ell p'_\ell = 1$, \eqref{App:StabReg:c} shows that $\mathcal{C}$ is closed and contains all convex combinations of the elements of $\mathcal{R}$. The desired result follows.

\subsection{Proof for Theorem~\ref{Theo:StabRegion}}\label{App:StabRegion}
Let $\tilde{\mathcal{C}}$ represent the polytope whose vertices belong to the set $\mathsf{ext}\ \mathcal{C}$ in \eqref{Eq:ExtPts}. The proof for $\tilde{\mathcal{C}}=\mathcal{C}$ comprises two steps: first, $\tilde{\mathcal{C}}$ is shown to be convex; second, the set $\mathcal{R}\backslash \mathsf{ext}\ \mathcal{C}$ is shown to be a subset of $\tilde{\mathcal{C}}$.

Proving the convexity of $\tilde{\mathcal{C}}$ is equivalent to showing that  its vertices are all \emph{exposed} \cite{UrrutyBook:ConvexAnalysis:04}. First, since all vectors in  $\tilde{\mathcal{C}}$ except for $\mathbf{0}$ are vectors of positive components, $\mathbf{0}$ is exposed by the \emph{supporting hyperplane} \cite{UrrutyBook:ConvexAnalysis:04} $\mathcal{H}_0 = \{\bv\in\mathds{R}^L: \langle\bv, \mathbf{1}\rangle=0\}$ where $\mathbf{1}$ represents the vector of all-$1$ components. Next, consider an arbitrary nonzero vertex $\bs$ of $\tilde{\mathcal{C}}$. We can write $\bs = d(m)\bv $ for some $m\in\mathcal{I}$ in \eqref{Eq:Index:Def} and $\bv\in\mathcal{V}_m$. Expand $\mathsf{ext}\ \mathcal{C}$ as $\mathsf{ext}\ \mathcal{C} := \{\bu_1, \bu_2, \cdots, \bu_T\}$ where $\bu_n\in\mathds{R}_+^L$ and $T = |\mathsf{ext}\ \mathcal{C}|$. An arbitrary point $\bw\in \tilde{\mathcal{C}}$ can be writhen as $\bw = \sum_{n = 1}^T a_n\bu_n$ where $a_n\geq 0 \ \forall \ n$ and $\sum_{n = 1}^T a_n = 1$. Thus $\langle\bw, \bv\rangle = \sum_{n = 1}^T a_n\langle\bu_n, \bv \rangle$. The term $\langle\bu_n, \bv\rangle$ belongs to one of the following four cases:
\begin{enumerate}
\item If  $\bu_n = d(m)\bv$, then $\langle\bu_n, \bv \rangle = md(m)$ since $\bv\in\mathcal{V}_m$;
\item If $\bu_n \in \mathcal{R}_m$ and $\bu_n \neq d(m)\bv$, then $\langle\bu_n, \bv \rangle < md(m)$;
\item If $\bu_n \in \mathcal{R}_a$ with $a> m$, then $\langle\bu_n, \bv \rangle \leq m d(a) < md(m)$;
\item If $\bu_n \in \mathcal{R}_a$ with $a< m$, then $\langle\bu_n, \bv \rangle \leq  a d(a) < md(m)$ where the last inequality follows from \eqref{Eq:Index:Def}.
\end{enumerate}
By combining these cases,   we conclude that $\langle\bu_n, \bv\rangle < md(m)$ and thus $\langle\bw, \bv\rangle < md(m)$ if $\bw \neq \bs$. In other words,  $\tilde{\mathcal{C}}$ lies in a half-space determined by the hyperplane $\mathcal{H}_{\bs} = \{\ba\in\mathds{R}^L: \langle \ba, \bv\rangle= md(m)\}$, and $\mathcal{H}_{\bs}\cap\tilde{\mathcal{C}} = \bs$. Thus $\bs$ is exposed by $\mathcal{H}_{\bs}$. Since $\bs$ is arbitrary, it follows that  all vertices of $\tilde{\mathcal{C}}$ are exposed, proving the convexity of  $\tilde{\mathcal{C}}$.

Define the complementary set of $\mathcal{I}$ in \eqref{Eq:Index:Def} as
\begin{equation}\label{Eq:Index:Comp}
\mathcal{I}^c = \l\{0\leq k \leq L : kd(k) \leq \max_{0\leq m < k}m d(m) \r\}.
\end{equation}
Note that $\mathcal{R}\backslash\mathsf{ext}\ \mathcal{C} = \bigcup_{k\in\mathcal{I}^c}\mathcal{R}_k$. Given that both $\tilde{\mathcal{C}}$ and $\mathcal{C}$ are convex polytopes, to prove $\tilde{\mathcal{C}} = \mathcal{C}$,  it is sufficient to show that the set $\bigcup_{k\in\mathcal{I}^c}\mathcal{R}_k\in \tilde{\mathcal{C}}$. Specifically, we show that an arbitrary point $\bd\in
\bigcup_{k\in\mathcal{I}^c}\mathcal{R}_k$ belongs to $\tilde{\mathcal{C}}$ as follows. There exists $a\in \mathcal{I}^c$ such that $\bd\in \mathcal{R}_{a}$ and $\bd = d(a)\hat{\bv}$ where $\hat{\bv}\in\mathcal{V}_{a}$. Define $\tilde{m} :=\arg \max_{0\leq m < a}\frac{md(m)}{a}$. From \eqref{Eq:Index:Comp}, $ad(a) \leq \tilde{m}d(\tilde{m})$. Furthermore, from \eqref{Eq:Index:Def}, $\tilde{m}\in\mathcal{I}$ and thus $\mathcal{R}_{\tilde{m}}\in\mathsf{ext}\ \mathcal{C}$. Next, define the set $\mathcal{W} := \{\bv\in\mathcal{V}_{\tilde{m}}: \langle\bv, \hat{\bv}\rangle = \tilde{m}\}$. Note that $\mathcal{W}$ is nonempty as $\hat{\bv}\in\mathcal{V}_a$ and $\bv\in\mathcal{V}_{\tilde{m}}$ and $a > \tilde{m}$. Moreover, define $\bz := \frac{1}{|\mathcal{W}|}\sum_{\bv\in\mathcal{W}}d(\tilde{m})\bv$ that is the convex combination of the points in $\mathcal{R}_{\tilde{m}}$, which are the vertices of $\tilde{\mathcal{C}}$. Since $\tilde{\mathcal{C}}$ is a convex polytope,  $\bz\in \tilde{\mathcal{C}}$. It follows that  the line segment $[\mathbf{0}, \bz]\in\tilde{\mathcal{C}}$. Note that $\bz = \frac{\tilde{m}}{a}d(\tilde{m})\hat{\bv}$. Since $d(a) \leq \frac{\tilde{m}}{a}d(\tilde{m})$, $\bd\in[\mathbf{0}, \bz]$ and thus $\bd\in\tilde{\mathcal{C}}$. Thus $\bigcup_{k\in\mathcal{I}^c}\mathcal{R}_k\in \tilde{\mathcal{C}}$. This completes the proof.

\subsection{Proof for Proposition~\ref{Prop:StabExtreme}}\label{App:StabExtreme}
Let $d_k(P)$ denote the same departure-rate function as $d_k(\bv)$ except with the different argument $P$. with the Consider $k > 1$ and $m < k$. Using Lemma~\ref{Lem:DepartRate} and L'H\^{o}spital's rule
\begin{eqnarray}
\lim_{P\rightarrow 0} \frac{d_k(P)}{d_m(P)} &=&  \lim_{P\rightarrow 0} \frac{\Gamma(L-m+1)\int_{\frac{k\theta}{P}}^\infty a^{L-k}e^{-a}da}{\Gamma(L-k+1)\int_{\frac{m\theta}{P}}^\infty b^{L-m}e^{-b}db}\nn\\
&=&  \lim_{P\rightarrow 0} \frac{\Gamma(L-m+1)\int_{\frac{k\theta}{P}}^\infty a^{L-k}e^{-a}da}{\Gamma(L-k+1)\int_{\frac{m\theta}{P}}^\infty b^{L-m}e^{-b}db}\nn\\
&=& \lim_{P\rightarrow 0} \frac{\Gamma(L-m+1)\l(\frac{k\theta}{P}\r)^{L-k}e^{-\frac{k\theta}{P}}}{\Gamma(L-k+1)\l(\frac{m\theta}{P}\r)^{L-m}e^{-\frac{m\theta}{P}}}\nn\\
&=& \lim_{P\rightarrow 0} \frac{\Gamma(L-m+1)\l(k\theta\r)^{L-k}}{\Gamma(L-k+1)\l(m\theta\r)^{L-m}}e^{(m-k)\theta\frac{1}{P}+o(\frac{1}{P})}\nn\\
&=& 0. \label{Prop:Proof:a}
\end{eqnarray}
It follows that the index set $\mathcal{I}$ defined in \eqref{Eq:Index:Def} converges with $P\rightarrow\infty$ as
\begin{equation}\label{Prop:Proof:b}
\lim_{P\rightarrow 0} \mathcal{I}(P) = \mathcal{R}_1\cup \mathcal{R}_0 = \mathcal{R}_1\cup \{\mathbf{0}\}.
\end{equation}
The first claim in the proposition statement follows from Theorem~\ref{Theo:StabRegion} and \eqref{Prop:Proof:b}.

Using Lemma~\ref{Lem:DepartRate} and Alzer's inequalities \cite{Alzer:GamFunIneq:97}, $d_k$ can be bounded as
\begin{equation}\label{Eq:Alzer:a}
 1 - \l(1- e^{-\frac{k\theta}{P} }\r)^{L-k} < d_k < 1 - \l(1- e^{-\frac{\beta k\theta}{P} }\r)^{L-k}
\end{equation}
where $\beta = [\Gamma(L-k+1)]^{\frac{1}{L-k}}$.  For $m < k$ and using \eqref{Eq:Alzer:a}, for $P \gg 1$
\begin{eqnarray}
d_k(P) - \frac{m}{k}d_m(P) &=& \l(1- e^{-\frac{m\theta}{P} }\r)^{L-m}- \l(1- e^{-\frac{\beta k\theta}{P} }\r)^{L-k}\nn\\
&=& \l(\frac{m\theta}{P}\r)^{L-m}- \l(\frac{\beta k\theta}{P} \r)^{L-k} + O\l(\l(\frac{1}{P}\r)^{2(L-k)}\r)\nn\\
&=& \l(\frac{\beta k\theta}{P} \r)^{L-k}\l[\frac{(m\theta)^{L-m}}{(\beta k\theta)^{L-k}}\l(\frac{1}{P}\r)^{k-m}- 1\r] + O\l(\l(\frac{1}{P}\r)^{2(L-k)}\r).\label{Prop:Proof:c}
\end{eqnarray}
From \eqref{Prop:Proof:c} and \eqref{Eq:Index:Def}, $\lim_{P\rightarrow\infty} \mathcal{I}(P)=\mathcal{R}$. Combining this result and Theorem~\ref{Theo:StabRegion} completes the proof for the second claim in the proposition statement.

\subsection{Proof for Lemma~\ref{Lem:QBeam:Sig}}\label{App:QBeam:Sig}
The result for $K=L$ follows from the fact that $\bff_n$ and $\bh_n$ are independent \cite{Jindal:MIMOBroadcastFiniteRateFeedback:06}. The proof for the case of $K < L$ is given as follows. Without loss of generality, consider the random variable $|\bff^\dagger_1\bh_1|^2$. Let $\mathcal{H}$ denote the $(K-1)$-dimension vector sub-space containing $(\hat{\bh}_2, \hat{\bh}_3,\cdots, \hat{\bh}_K)$. A basis of $\mathcal{H}$ comprises the vectors $(\bee_1, \bee_2, \cdots, \bee_{K-1})$. The null space of $\mathcal{H}$, represented by $\mathcal{H}^c$, is spanned by the basis vectors  $(\bee_K, \bee_{K+1}, \bee_L)$. Given these definitions, the unitary vector $\hat{\bs}_1 := \hat{\bh}_1/|\hat{\bh}_1| $ can be decomposed as $\hat{\bs}_1 = \sum_{\ell=1}^{L}  \hat{s}_{1,\ell}\bee_\ell$ where $\sum_{\ell=1}^L|\hat{s}_{1,\ell}|^2=1$. For zero-forcing beamforming discussed in Section~\ref{Section:Beamform}
\begin{equation}
\left|\bff_1^\dagger \hat{\bs}_1\right|^2 = \max_{\bff\in\mathcal{H}^c}\left|\bff^\dagger\times \sum_{\ell=1}^{L}  \hat{s}_{1, \ell}\bee_\ell\right|^2 =\max_{\bff\in\mathcal{H}^c}\left|\bff^\dagger\times \sum_{\ell=K}^{L}  \hat{s}_{1, \ell}\bee_\ell\right|^2.\nn
\end{equation}
It follows that
\begin{equation}\label{Eq:BF}
\bff_1 = \frac{1}{\sqrt{\sum_{\ell=K}^L |\hat{s}_{1, \ell}|^2 }}\sum_{\ell=K}^L \hat{s}_{1, \ell} \bee_\ell.
\end{equation}
Define $\bs_1 := \bh_1/\|\bh_1\|$ and decompose it as $\bs_1= \sum_{\ell=1}^{L}  s_{1,\ell}\bee_\ell$. The projection of $\bs_1$ into the subspace $\mathcal{H}^c$ is hence $\tilde{\bs}_1 := \sum_{\ell=K+1}^{L}s_{1,\ell}\bee_\ell$. Moreover, decompose $\bh_1$ as $\bh_1 = \sum_{\ell=1}^L h_{1,\ell}\bee_\ell$. Similarly, the projection of $\bh_1$ into $\mathcal{H}^c$ is given as $\tilde{\bh}_1 = \sum_{m=K+1}^L h_{1,m}\bee_m$.
Using above definitions and \eqref{Eq:BF}
\begin{equation}
\left|\bh_1^\dagger \bff_1\right|^2  = \frac{\|\bh_1\|^2}{\sum_{\ell=K}^L |\hat{s}_{1, \ell}|^2 }|\tilde{\bs}_1^\dagger\hat{\bs}_1|^2
\overset{(a)}{\geq} \|\bh_1\|^2|\tilde{\bs}_1^\dagger\hat{\bs}_1|^2.\label{Eq:App:c}
\end{equation}
The inequality $(a)$ follows from $\sum_{\ell=K}^L |\hat{s}_{1, \ell}|^2 \leq 1$. Based on quantized CSI model in Section~\ref{Section:LimFb} and by geometry \cite{MukSabETAL:BeamFiniRateFeed:Oct:03}
\begin{equation}\label{Eq:App:e}
|\tilde{\bs}_1^\dagger\hat{\bs}_1|^2 = (1-\epsilon_1) \frac{\|\tilde{\bh}_1\|^2}{\|\bh_1\|^2} + \rho_1 \frac{\|\bh_1\|^2-\|\tilde{\bh}_1\|^2}{\|\bh_1\|^2}
\end{equation}
where $-\epsilon_1\leq \rho_1\leq \epsilon_1$. From \eqref{Eq:App:c} and \eqref{Eq:App:e}
\begin{equation}
\left|\bh_1^\dagger \bff_1\right|^2  \geq (1-\epsilon_1) \|\tilde{\bh}_1\|^2 + \rho (\|\bh_1\|^2-\|\tilde{\bh}_1\|^2). \label{Eq:App:f}
\end{equation}
Given that $\{h_{1,m}\}$ are i.i.d. exponential random variables with unit mean, the desired result follows from \eqref{Eq:App:f}.

\subsection{Proof for Lemma~\ref{Lem:Loss:Bnds}}\label{App:Loss:Bnds}
Let $\mathcal{A}$ denote the set of indices of the $K$ scheduled mobiles and nonempty queues. To simplify notation, define $T_{\Sigma} = \sum_{\substack{m\in\mathcal{A}\\ m\neq \ell}} T_{\ell, m}$. Since $\{T_{\ell, m}\}$ are i.i.d. $\beta(1, L-1)$ random variables from Lemma~\ref{Lem:Interf},  $T_{\Sigma}\leq K-1$ and $\E[T_{\Sigma}] = 1$. Using Lemma~\ref{Lem:QBeam:Sig}, Lemma~\ref{Lem:Interf} and \eqref{Eq:SINR:Def}, the SINR can be rewritten as
\begin{eqnarray}
\SINR_\ell &=& \frac{\gamma (aZ_\ell+\rho_\ell Q_\ell)}{1+\gamma \epsilon_\ell (Z_\ell+Q_\ell)T_{\Sigma}}\nn\\
&\geq& \frac{\gamma [(1-\epsilon_\ell)Z_\ell-\epsilon_\ell Q_\ell]}{1+\gamma \epsilon_\ell T_{\Sigma}(Z_\ell+Q_\ell)}.\label{Eq:SINR}
\end{eqnarray}
From \eqref{Eq:SINR}, the departure rate $\hat{d}(K)$ is lower bounded as
\begin{eqnarray}
\hat{d}(K) &\geq& \Pr\(\frac{\gamma [(1-\epsilon_\ell)Z_\ell-\epsilon_\ell Q_\ell]}{1+\gamma T_{\Sigma}\epsilon_\ell (Z_\ell+Q_\ell)}\geq \theta \)\nn\\
&=& \E[\Pr(Z \geq \underbrace{\frac{\theta/\gamma + \epsilon_\ell Q_\ell(1+T_{\Sigma}\theta)}{1-\epsilon_\ell(1+T_{\Sigma}\theta)}}_{\Pi}\mid Q_\ell)]\nn\\
&\overset{(a)}{=}& \E\l[\sum\nolimits_{n=0}^{L-K}\Pi^n\exp\l(-\Pi\r)\r]\nn\\
&>& \E\l[\sum\nolimits_{n=0}^{L-K}(\theta/\gamma)^n\exp\l(-\Pi\r)\r]\nn\\
&=& \E\l[\exp\l(-\Pi+\theta/\gamma\r)\r]\sum\nolimits_{n=0}^{L-K}(\theta/\gamma)^ne^{-\theta/\gamma}\nn\\
&\geq& \E\l[\exp\l(-\frac{\epsilon_\ell(1+T_{\Sigma}\theta)(Q_\ell+\theta/\gamma)}{1-\epsilon_\ell(1+T_{\Sigma}\theta)}\r)\r]d(K)\nn\\
&\overset{(b)}{\geq}& \l[1-\frac{\E[\epsilon_\ell](1+T_{\Sigma}\theta)(\E[Q_\ell]+\theta/\gamma)}{1-2^{-\frac{B}{L-1}}(1+(K-1)\theta)}\r]d(K)\nn\\
&\overset{(c)}{=}& [1-\underbrace{\frac{2^{-\frac{B}{L-1}}\l(1-\frac{1}{L}\r)(1+\theta)(K-1+\theta/\gamma)}{1-2^{-\frac{B}{L-1}}(1+(K-1)\theta)}}_{f(B)}]d(K).\nn
\end{eqnarray}
The equality (a) holds since $Z_\ell$ is a $\chi(L-K+1)$ random variable according to Lemma~\ref{Lem:Interf}. The inequality (b) is obtained using $\epsilon_\ell \leq 2^{-\frac{B}{L-1}}$ and  $T_{\Sigma}\leq K-1$. The equality (c) follows from $\E[\epsilon_\ell] = 2^{-\frac{B}{L-1}}\l(1-\frac{1}{L}\r)$, $\E[T_{\Sigma}]=1$ and $\E[Q_\ell] = K-1$.
By setting $f(B) \leq \delta$, we obtain that
\begin{equation}
2^{-\frac{B}{L-1}} \leq \frac{\delta}{ \l(1-\frac{1}{L}\r)(1+\theta)\l(K-1+\frac{\theta K}{P}\r)+\delta(1+(K-1)\theta)}.\nn
\end{equation}
Let $K=L$ and we have more stringent feedback requirements
\begin{equation}
2^{-\frac{B}{L-1}} \leq \frac{\delta}{ \l(L-1\r)(1+\theta)\l(1-\frac{1}{L}+\frac{\theta}{P}\r)+\delta(1+(L-1)\theta)}.\nn
\end{equation}
The desired result follows from the above inequality.

\subsection{Proof for Proposition~\ref{Prop:Delay:Poisson}}\label{App:Delay:Poisson}
The average queueing delay $W$ is given by \eqref{Eq:P-K:a} but with $\hat{\mu}$ replaced by $\mu$. From \eqref{Eq:P-K:a} and \eqref{Eq:LossFact}
\begin{eqnarray}
\frac{W}{\widehat{W}} &\overset{(a)}{\leq}& \frac{1}{1-\delta}\times\frac{\mu-\lambda}{\hat{\mu}-\lambda}\times \frac{2-\hat{\mu}}{2-\mu}\nn\\
&<& \frac{1}{1-\delta}\times\frac{\mu-\lambda}{\hat{\mu}-\lambda}\nn\\
&=& \frac{\tau}{(1-\delta)(\tau-\delta)}.\label{Eq:App:b}
\end{eqnarray}
where (a) holds since $\hat{\mu}/\mu\geq 1-\delta$ from \eqref{Eq:LossFact} and the definitions of $\mu$ and $\hat{\mu}$. 
Therefore a sufficient condition for $W/\widehat{W}\leq M$ is
\begin{equation}
\frac{\tau}{(1-\delta)(\tau-\delta)}\leq M\nn
\end{equation}
or equivalently
\begin{equation}
\delta^2 - (1+\tau)\delta + \l(1-\frac{1}{M}\r)\tau \geq 0.
\end{equation}
It follows that
\begin{eqnarray}
\delta &\leq& \frac{1+\tau-\sqrt{(1+\tau)^2-4\l(1-\frac{1}{M}\r)\tau}}{2}\nn\\
&=& \frac{1+\tau}{2}\l[1-\sqrt{1-\frac{4\l(1-\frac{1}{M}\r)\tau}{(1+\tau)^2}}\r].\label{eq:App:d}
\end{eqnarray}
One can check that the right-hand-side of \eqref{eq:App:d} is always real for $M>1$. The desired result follows from \eqref{eq:App:d} and Lemma~\ref{Lem:Loss:Bnds}.

\subsection{Proof for Lemma~\ref{Lem:Delay:General}}\label{App:Delay:General}
Given ARQ transmission, the moment generating function of $\hat{Y}$ is obtained as
\begin{eqnarray}
\E[e^{r\hat{Y}}] &=& \sum_{n=1}^\infty \hat{\mu}(1-\hat{\mu})^{n-1} e^{rn}\nn\\
&=& \frac{\hat{\mu}e^r}{1-(1-\hat{\mu})e^r}. \label{Eq:CharFunc}
\end{eqnarray}
Since $\hat{Y}$ and $X$ are independent, by substituting \eqref{Eq:CharFunc} into \eqref{Eq:ExpEq}
\begin{equation}
\hat{\mu}\E[e^{-\hat{r}^\star X}] - e^{-\hat{r}^\star} + 1-\hat{\mu} = 0.\label{Eq:ExpEq:a}
\end{equation}
We can write $ \hat{\mu} =(1-\sigma)\mu$ with $0<\sigma<1$. Note that $\sigma <\delta$ with $\delta$ given in \eqref{Eq:LossFact}. Moreover, define $r^\star$ as the solution of the following perfect-CSIT counterpart of \eqref{Eq:ExpEq:a}
\begin{equation}
\mu\E[e^{-r^\star X}] - e^{-r^\star} + 1-\mu = 0.\label{Eq:ExpEq:b}
\end{equation}
For $\delta\rightarrow 0$ and hence $\sigma\rightarrow 0$, the asymptotic relationship between $\hat{r}^\star$ and $r^\star$ is obtained using \emph{perturbation theory} \cite{SimBook:FirstLookPerturbationTheory:97} as follows. Expand $\hat{r}^\star$ as $\hat{r}^\star := \alpha + \beta\sigma + O(\sigma^2)$. By substituting this expression into \eqref{Eq:ExpEq:a}
\begin{equation}
\underbrace{\mu\E[e^{-\alpha X}] - e^{-\alpha} + 1 - \mu}_{\Pi_1} - \underbrace{\l\{\mu\E[e^{-\alpha X}X]\beta + \mu\E[e^{-\alpha X}] - \beta e^{-\alpha}-\mu\r\}}_{\Pi_2} \sigma+ O(\sigma^2) = 0
\end{equation}
Following the approach of perturbation theory, the terms $\Pi_1$ and $\Pi_2$ are set equal to zero, thus
\begin{eqnarray}
\mu\E[e^{-\alpha X}] - e^{-\alpha} + 1 - \mu &=& 0\nn\\
\mu\E[e^{-\alpha X}X]\beta + \mu\E[e^{-\alpha X}] - \beta e^{-\alpha}-\mu &=& 0.\nn
\end{eqnarray}
By solving the above equations, we obtain that $\alpha = r^\star$ and $\beta :=-f(r^\star)$ with $f(r^\star)$ in \eqref{Eq:Beta}. This completes the proof.

\bibliographystyle{ieeetr}

\begin{thebibliography}{}

\end{thebibliography}


\begin{thebibliography}{10}

\bibitem{Gesbert:ShiftMIMOParadigm:2007}
D.~Gesbert, M.~Kountouris, R.~W.~H. Jr., C.-B. Chae, and T.~Salzer, ``From
  single user to multiuser communications: Shifting the {MIMO} paradigm,'' {\em
  IEEE Signal Proc. Magazine}, vol.~24, pp.~36--46, Sept. 2007.

\bibitem{Caire:AchivThroghputBroadcastChan:2003}
G.~Caire and S.~Shamai, ``On the achievable throughput of a multiantenna
  {Gaussian} broadcast channel,'' {\em IEEE Trans. on Info. Theory}, vol.~49,
  no.~7, pp.~1691--1706, 2003.

\bibitem{Vishwanath:DualityAchRatesBroadcastChan:2003}
S.~Vishwanath, N.~Jindal, and A.~Goldsmith, ``Duality, achievable rates, and
  sum-rate capacity of {Gaussian} {MIMO} broadcast channels,'' {\em IEEE Trans.
  on Info. Theory}, vol.~49, pp.~2658--68, Oct. 2003.

\bibitem{YooGoldsmith:OptimBroadcastZeroForcingBeam:2006}
T.~Yoo and A.~Goldsmith, ``On the optimality of multiantenna broadcast
  scheduling using zero-forcing beamforming,'' {\em IEEE Journal on Sel. Areas
  in Communications}, vol.~24, pp.~528--541, Mar. 2006.

\bibitem{Jindal:MIMOBroadcastFiniteRateFeedback:06}
N.~Jindal, ``{MIMO} broadcast channels with finite-rate feedback,'' {\em IEEE
  Trans. on Info. Theory}, vol.~52, pp.~5045--5060, Nov. 2006.

\bibitem{Huang:OrthBeamSDMALimtFb:07}
K.-B. Huang, J.~G. Andrews, and R.~W. Heath~Jr., ``Performance of orthogonal
  beamforming for {SDMA} systems with limited feedback,'' {\em to appear in
  IEEE Trans. on Veh. Technology}.

\bibitem{SpencerSwindleETAL:ZFsdma:2004}
Q.~H. Spencer, A.~L. Swindlehurst, and M.~Haardt, ``Zero-forcing methods for
  downlink spatial multiplexing in multiuser {MIMO} channels,'' {\em IEEE
  Trans. on Signal Processing}, vol.~52, no.~2, pp.~461 -- 471, 2004.

\bibitem{SwannackWornell:FindingNEMO:2005}
C.~Swannack, E.~Uysal-Biyikoglu, and G.~W. Wornell, ``Finding {NEMO}: Near
  mutually orthogonal sets and applications to {MIMO} broadcast scheduling,''
  in {\em Proc. Int. Conf. Wireless Networks, Commun., Mobile Computing}, June
  2005.

\bibitem{YooJindal:FiniteRateBroadcastMUDiv:2007}
T.~Yoo, N.~Jindal, and A.~Goldsmith, ``Multi-antenna broadcast channels with
  limited feedback and user selection,'' {\em IEEE Journal on Sel. Areas in
  Communications}, vol.~25, pp.~1478--1491, July 2007.

\bibitem{SharifHassibi:CapMIMOBroadcastPartSideInfo:Feb:05}
M.~Sharif and B.~Hassibi, ``On the capacity of {MIMO} broadcast channels with
  partial side information,'' {\em IEEE Trans. on Info. Theory}, vol.~51,
  pp.~506--522, Feb. 2005.

\bibitem{Love:OverviewLimitFbWirelssComm:2008}
D.~J. Love, R.~W. Heath, V.~K.~N. Lau, D.~Gesbert, B.~D. Rao, and M.~Andrews,
  ``An overview of limited feedback in wireless communication systems,'' {\em
  IEEE Journal on Sel. Areas in Communications}, vol.~26, no.~8,
  pp.~1341--1365, 2008.

\bibitem{SwannackWorenell:BroadcastChanLimitedCSI:2005}
C.~Swannack, E.~Uysal-Biyikoglu, and G.~W. Wornell, ``{MIMO} broadcast
  scheduling with limited channel state information,'' in {\em Proc., Allerton
  Conf. on Comm., Control, and Computing}, Sept. 2005.

\bibitem{Telatar:CombineQueueTheoInfoTheory:1995}
I.~E. Telatar and R.~G. Gallager, ``Combining queueing theory with information
  theory for multiaccess,'' {\em IEEE Journal on Sel. Areas in Communications},
  vol.~13, pp.~963--969, Aug. 1995.

\bibitem{YehThesis:MACFadingCommNeworks:2001}
E.~Yeh, {\em Multiaccess urd Fading in Communication Networks}.
\newblock Ph.D. Thesis, EECS, MIT, 2001.

\bibitem{YehCohen:ThputDelayOptimalResourceAlloacMAC:2003}
E.~M. Yeh and A.~S. Cohen, ``Throughput and delay optimal resource allocation
  in multiaccess fading channels,'' in {\em Proc., IEEE Intl. Symposium on
  Information Theory}, pp.~245--245, June/July 2003.

\bibitem{CoverBook}
T.~M. Cover and J.~A. Thomas, {\em Elements of Information Theory}.
\newblock John Wiley, 1991.

\bibitem{GeorgNeelyBook}
L.~Georgiadis, M.~Neely, and L.~Tassiulas, {\em Resource Allocation and Cross
  Layer Control in Wireless Networks}.
\newblock Now Publishers Inc, 1~ed., 2006.

\bibitem{YehCohen:ThputOptimPowrRateControlMACBC:2004}
E.~M. Yeh and A.~S. Cohen, ``Throughput optimal power and rate control for
  queued multiaccess and broadcast communications,'' in {\em Proc., IEEE Intl.
  Symposium on Information Theory}, June/July 2004.

\bibitem{Neely:PowrAllocationRoutMultibeamSatellites:2003}
M.~J. Neely, E.~Modiano, and C.~E. Rohrs, ``Power allocation and routing in
  multibeam satellites with time-varying channels,'' {\em IEEE Trans. on
  Networking}, vol.~11, pp.~138--152, Feb. 2003.

\bibitem{SwannackWorenell:LowComplexMuScheduleMIMOBC:2004}
C.~Swannack, E.~U. Biyikoglu, and G.~W. Wornell, ``Low complexity multiuser
  scheduling for maximizing throughput in the {MIMO} broadcast channel,'' in
  {\em Proc., Allerton Conf. on Comm., Control, and Computing}, Sept. 2004.

\bibitem{VisKum:RateScheduleMultiAntDLSys:2005}
H.~Viswanathan and K.~Kumaran, ``Rate scheduling in multiple antenna downlink
  wireless systems,'' {\em IEEE Trans. on Communications}, vol.~53,
  pp.~645--655, Apr. 2005.

\bibitem{BertsekasBook:DataNetwk:92}
D.~Bertsekas and R.~Gallager, {\em Data networks}.
\newblock Prentice Hall, 1992.

\bibitem{GallagerBook:StochasticProcs:95}
R.~Gallager, {\em Discrete Stochastic Processes}.
\newblock Springer, 1995.

\bibitem{SimBook:FirstLookPerturbationTheory:97}
J.~G. Simmonds and J.~E. Mann, {\em A First Look at Perturbation Theory}.
\newblock Dover Publications, 1997.

\bibitem{Tel:CapaMultGausChan:99}
I.~E. Telatar, ``Capacity of multi-antenna {Gaussian} channels,'' {\em European
  Trans. on Telecomm.}, vol.~10, no.~6, pp.~585--595, 1999.

\bibitem{TarokhJafETAL:SpacBlocCodeFrom:Jul:99}
V.~Tarokh, H.~Jafarkhani, and A.~R. Calderbank, ``Space-time block codes from
  orthogonal designs,'' {\em IEEE Trans. on Info. Theory}, vol.~45,
  pp.~1456--1467, Jul. 1999.

\bibitem{CaireETAL:PwrCtrlFadingChan:99}
G.~Caire, G.~Taricco, and E.~Biglieri, ``Optimum power control over fading
  channels,'' {\em IEEE Trans. on Inform. Theory}, vol.~45, pp.~1468--89, July
  1999.

\bibitem{Jindal:RethinkMIMONetwork:LinearThroughput:2008}
N.~Jindal, J.~G. Andrews, and S.~Weber, ``Rethinking {MIMO} for wireless
  networks: Linear throughput increases with multiple receive antennas,'' {\em
  submitted: Proc., IEEE Intl. Conf. on Communications}, Sept. 2008.

\bibitem{LovHeaETAL:Gras:May:03}
D.~J. Love, R.~W. Heath~Jr., and T.~Strohmer, ``Grassmannian beamforming for
  multiple-input multiple-output wireless systems,'' in {\em Proc. of IEEE Int.
  Conf. on Commun.}, vol.~4, pp.~2618--2622, May 2003.

\bibitem{Zhou:QuantifyPowrLossTxBeamFiniteRateFb:2005}
S.~Zhou, Z.~Wang, and G.~B. Giannakis, ``Quantifying the power loss when
  transmit beamforming relies on finite-rate feedback,'' {\em IEEE Trans. on
  Communications}, vol.~4, pp.~1948--1957, July 2005.

\bibitem{GerGra:VectQuanSignComp:92}
A.~Gersho and R.~M. Gray, {\em Vector Quantization and Signal Compression}.
\newblock Kluwer Academic Press, 1992.

\bibitem{LovHeaETAL:GrasBeamMultMult:Oct:03}
D.~J. Love, R.~W. Heath~Jr., and T.~Strohmer, ``Grassmannian beamforming for
  multiple-input multiple-output wireless systems,'' {\em IEEE Trans. on Info.
  Theory}, vol.~49, pp.~2735--2747, Oct. 2003.

\bibitem{Szpankowski:StabCondMultiQueue:1994}
W.~Szpankowski, ``Stability conditions for some multi-queue distributed
  systems: Buffered random access systems,'' {\em Adv. Appl. Probab.}, vol.~26,
  p.~498–515, June 1994.

\bibitem{BaccelliBook}
F.~Baccelli and P.~Bremaud, {\em Elements of Queueing Theory}.
\newblock Springer, 2~ed., 2003.

\bibitem{Naware:StableDelayFiniteUserALOHA:2005}
V.~Naware, G.~Mergen, and L.~Tong, ``Stability and delay of finite-user slotted
  {ALOHA} with multipacket reception,'' {\em IEEE Trans. on Info. Theory},
  vol.~51, pp.~2636--2656, July 2005.

\bibitem{UrrutyBook:ConvexAnalysis:04}
C.~L. J.-B. Hiriart-Urruty, {\em Fundamentals of Convex Analysis}.
\newblock Springer, 2004.

\bibitem{SharifHassibi:DelayConsiderOppScheduleBC:2007}
M.~Sharif and B.~Hassibi, ``Delay considerations for opportunistic scheduling
  in broadcast fading channels,'' {\em IEEE Trans. on Wireless Communications},
  vol.~6, pp.~3353--3363, Sept. 2007.

\bibitem{Alzer:GamFunIneq:97}
H.~Alzer, ``On some inequalities for the incomplete {Gamma} function,'' {\em
  Mathematics of Computation}, vol.~66, pp.~771--778, Apr. 2005.

\bibitem{MukSabETAL:BeamFiniRateFeed:Oct:03}
K.~K. Mukkavilli, A.~Sabharwal, E.~Erkip, and B.~Aazhang, ``On beamforming with
  finite rate feedback in multiple antenna systems,'' {\em IEEE Trans. on Info.
  Theory}, vol.~49, pp.~2562--79, Oct. 2003.

\end{thebibliography}


\end{document}